\documentclass[11pt]{article}

\usepackage[letterpaper,margin=1in]{geometry}

\usepackage[utf8]{inputenc} % allow utf-8 input
\usepackage[T1]{fontenc}    % use 8-bit T1 fonts
\usepackage{url}            % simple URL typesetting
\usepackage{booktabs}       % professional-quality tables
\usepackage{amsfonts}       % blackboard math symbols
\usepackage{nicefrac}       % compact symbols for 1/2, etc.
\usepackage{microtype}      % microtypography
\usepackage[table,xcdraw]{xcolor}         % colors
\usepackage{tikz}
\usepackage{authblk}

\usepackage[symbol]{footmisc}

\usepackage{algorithm}
\usepackage[noend]{algpseudocode}
\usepackage{nicefrac,bm}
\algnewcommand{\IIf}[1]{\State\algorithmicif\ #1\ \algorithmicthen}
\algnewcommand{\EndIIf}{\unskip\ \algorithmicend\ \algorithmicif}
\usepackage{amsmath,amsthm,amssymb,mathrsfs,color,colortbl,xspace}

\usepackage{enumitem}
\usepackage{caption,graphicx,subcaption}
\usepackage{natbib}
\usepackage{wrapfig}
\usepackage{comment}
\usepackage{thmtools} 
\usepackage{thm-restate}
\usepackage[hidelinks]{hyperref}
\usepackage[capitalize]{cleveref}

\newcommand{\E}[1]{\mathbb{E}\left[#1\right]}
\newcommand{\Eo}[1]{\mathbb{E}^1\left[#1\right]}

\newcommand{\Ez}[1]{\mathbb{E}^0\left[#1\right]}
\newcommand{\Ei}[1]{\mathbb{E}^i\left[#1\right]}

\newcommand{\eps}{\varepsilon}
\renewcommand{\P}[1]{\mathbb{P}\left(#1\right)}
\newcommand{\Po}[1]{\mathbb{P}^1\left(#1\right)}
\newcommand{\Pt}[1]{\mathbb{P}^2\left(#1\right)}
\renewcommand{\Pi}[1]{\mathbb{P}^i\left(#1\right)}
\newcommand{\Pz}[1]{\mathbb{P}^0\left(#1\right)}
\newcommand{\bno}[1]{\bigl\lVert#1\bigr\rVert}
\newcommand{\babs}[1]{\bigl\lvert#1\bigr\rvert}
\newcommand{\kl}{\mathcal{D}_{\mathrm{KL}}}
\newcommand{\Pb}{\mathbb{P}}

\newcommand{\A}{\mathcal{A}}

\renewcommand{\S}{\mathcal{S}}
\newcommand{\U}{\mathcal{U}}

\newcommand{\supp}{\textsc{Supp}\xspace}
\newcommand{\calE}{\mathcal{E}}

\newcommand{\qminus}{q^-}
\newcommand{\qplus}{q^+}
\newcommand{\ind}[1]{\mathbb{I}{\left\{#1\right\}}}

\newcommand{\N}{\mathbb{N}}

\newcommand{\sw}{\textsc{SW}}

\newcommand{\gft}{\textsc{GFT}}
\newcommand{\egft}{\widehat \gft}

\newcommand{\blockdec}{\textsc{Block-Decomposition}\xspace}
\newcommand{\bd}{\textsc{BD}}

\DeclareMathOperator*{\argmax}{arg\,max}

\usepackage{thmtools, thm-restate}

\newcommand{\new}[1]{{#1}}
\newtheorem{theorem}{Theorem}
\newtheorem{lemma}{Lemma}
\newtheorem{claim}{Claim}

\newtheorem{observation}{Observation}

\title{An \texorpdfstring{$\alpha$-regret\xspace}\ analysis of Adversarial Bilateral Trade}

\author[1]{Yossi Azar}
\author[1]{Amos Fiat}
\author[2]{Federico Fusco}

\affil[1]{Tel Aviv University, Israel}
\affil[2]{Sapienza Università di Roma, Italy}

\begin{document}

\maketitle

\begin{abstract}
    We study sequential bilateral trade where sellers and buyers valuations are completely arbitrary ({\sl i.e.}, determined by an adversary). Sellers and buyers are strategic agents with private valuations for the good and the goal is to design a mechanism that maximizes efficiency (or gain from trade) while being incentive compatible, individually rational and budget balanced. In this paper we consider gain from trade, which is harder to approximate than social welfare.
    
    We consider a variety of feedback scenarios and distinguish the cases where the mechanism posts one price and when it can post different prices for buyer and seller. We show several surprising results about the separation between the different scenarios. In particular we show that (a) it is impossible to achieve sublinear $\alpha$-regret for any $\alpha<2$, (b) but with full feedback sublinear $2$-regret is achievable; (c) with a single price and partial feedback one cannot get sublinear $\alpha$ regret for any constant $\alpha$  (d) nevertheless, posting two prices even with one-bit feedback achieves sublinear $2$-regret, and (e) there is a provable separation in the $2$-regret bounds between full and partial feedback. 
\end{abstract}

% Paper body

\section{Introduction}

    The bilateral trade problem arises when two rational agents, a seller and a buyer, wish to trade a good; they both hold a private valuation for it, and their goal is to maximize their utility. This scenario emerges naturally in many applications, such as ridesharing systems (e.g., Uber or Lyft), where trades between sellers/drivers and buyers/riders are managed by an intermediating platform. The solution to the bilateral trade problem consists in designing a mechanism that intermediates between the two parties to make the trade happen. 
    
    Ideally, the mechanism should maximize social welfare even though the agents act strategically ({\sl incentive compatibility}) and should guarantee non-negative utility to the agents ({\sl individual rationality}); moreover, we are interested in mechanisms for bilateral trade that do not subsidize the agents ({\sl budget balance}). Obvious mechanisms that satisfy incentive compatibility, individual rationality, and budget balanced, are posted price mechanisms. 
    Two common metrics are used to measure the efficiency of a mechanism: social welfare subsequent to trade and gain from trade ({\sl i.e.}, the increase in social welfare). Consider a mechanism that posts prices $p$ (price for the seller) and $q$ (price for the buyer) to agents with valuations $s$ and $b$, formally we have:
    \begin{itemize}
        \item \textbf{Social Welfare}: $\sw(p,q,s,b) = s + (b-s) \cdot \ind{s \le p\le q \le b}$\footnote{We use $\ind{Q}$ for the indicator variable that takes the value $1$ if the predicated $Q$ is true and zero otherwise.}
        \item \textbf{Gain from trade}: $\gft(p,q,s,b) = (b-s)\cdot \ind{s \le p\le q \le b}$
    \end{itemize}
     It is clear from these expressions that if we are interested in exact optimality then maximizing gain from trade is equivalent to maximizing social welfare.
     However, \citet{Myerson83} showed that there are {\sl no mechanisms} for bilateral trade that are simultaneously social welfare  maximizing (alternately, gain from trade maximizing), incentive compatible, individually rational, and budget balanced\footnote{This impossibility result holds even when the (private) agents valuations are assumed to be drawn from some (public) random distributions and the incentive compatibility is only enforced in expectation.}.
     It follows that the best one can hope for is an incentive compatible, individually rational, and budget balanced mechanism that \textit{approximates} the optimal social welfare. This creates an asymmetry between the two metrics:
    a multiplicative $c$ approximation to the maximal gain from trade implies an approximation at least as good ($\geq c$) to the maximal social welfare but not vice versa. Ergo, it is harder to approximate the gain from trade than to approximate social welfare. For example, consider an instance where the seller has valuation $0.99$ and the buyer is willing to pay up to $1$: irrespective of whether a trade occurs or  not, $99\%$ of the optimal social welfare is guaranteed. In particular, a mechanism that posts a price of zero and generates no trade gets a good approximation to the social welfare. Contrariwise, the gain from trade is non-zero only if the mechanism manages to post prices in the narrow $[0.99,1]$ interval. 

     The vast body of work subsequent to \citet{Myerson83} primarily considers the Bayesian version of the problem, where agents' valuations are drawn from some distribution and the efficiency is evaluated in expectation with respect to the valuations' randomness. There are many incentive compatible mechanisms that give a constant approximation to the social welfare \citep[see, e.g.,][and previous works]{CaiW23,LiuR023}.  
     On the other hand, finding a constant approximation to the gain from trade has been a long standing problem and only a recent paper of \citet{DengMSW21} has given a Bayesian incentive compatible mechanism for this problem. In this paper we deal with the harder scenario where an adversary determines seller and buyer valuations ({\sl i.e.},  valuations are not drawn from some distribution). Ergo, positive results in the Bayesian setting are inapplicable in our setting. 
     
     Following \citet{Nicolo21}, we consider the sequential adversarial bilateral trade problem, where at each time step $t$, a new seller-buyer pair arrives. The seller has some private valuation $s_t \in [0,1]$ representing the smallest price she is willing to accept; conversely, the buyers holds as private information $b_t \in [0,1]$, {\sl i.e.}, the largest price she is willing to pay to get the good. Concurrently, the mechanism posts price $p_t$ to the seller and $q_t$ to the buyer. If they both accept ($s_t \le p_t$ and $q_t \le b_t$), then the trade happens at those prices, otherwise the agents leave forever. By the requirement that the mechanism be budget balanced, the prices posted by the mechanism are such that $p_t \le q_t.$ At the end of each time step, the mechanism receives some feedback that depends on the outcome of the trade. Ideally, we would like to have a strategy for the sequential bilateral trade problem whose average gain from trade converges to that of the best fixed posted price mechanism in hindsight. However, as \citet{Nicolo21} showed, this is a hopeless task. 
     
     Our goal in this work is then to achieve mechanisms whose average performance converges to a constant factor of the best fixed posted price mechanism in hindsight. We would like to find the smallest $\alpha \ge 1$ such that the $\alpha$-regret \citep{KakadeKL09} is sublinear in the time horizon $T$\new{. Formally, we would like the following to hold:}
    \[
        \max_{p,q}\sum_{t=1}^T \gft(p,q,s_t,b_t) - \alpha \cdot \E{\sum_{t=1}^T \gft(p_t,q_t,s_t,b_t)} \new{\in o(T)} .
    \]
    If the goal is only to maximize gain from trade, there is never any sense in offering two different prices (to the seller and buyer). However, critically, offering two prices is provably helpful in the context of a learning algorithm. 
    
    To conclude the description of our learning framework, we specify the type of feedback received by the mechanism. We focus on the two extremes of the feedback spectrum. On the one hand we study the full feedback model, where, after prices are posted, the mechanism learns both seller and buyer valuations $(s_t,b_t)$. On the other hand, we investigate a more realistic {\sl partial feedback} model, the {\em one-bit feedback}, where the learner only discovers if a trade took place or not. We also consider an intermediate (partial feedback) model, called the {\sl two-bit feedback model}. In this model, the learner posts (one or two) prices, and learns if the buyer is willing to trade and if the seller is willing to trade, at these prices. Clearly, a trade actually occurs only if both are willing to trade. Note that these two models enforce the desirable property that buyers and sellers only communicate to the mechanism a minimal amount of information useful for the trade, without disclosing their actual valuations. 
    \new{We conclude with a final observation concerning the adversarial model. Differently from the stochastic setting, where the valuations are drawn i.i.d. from a fixed but unknown distribution, the adversarial setting that we study models possibly non-stationary data generation processes. For this reason, simply using past observations to estimate the gain from trade function and infer the next prices to post does not work.}

    \begin{table}[t!]
        \centering
        \renewcommand{\arraystretch}{1.2}
        \begin{tabular}{c|c|c|c|}
        \cline{2-4}
        & \multicolumn{1}{l|}{Full Feedback} & \multicolumn{1}{l|}{Two-bit feedback} & One-bit feedback                   \\ \hline
        \multicolumn{1}{|l|}{Single price} & \cellcolor[HTML]{32CB00}$O(\sqrt T)$ - Thm. \ref{thm:upper-full}                  & $\cellcolor[HTML]{FD6864}{\Omega(T)}$ - Thm. \ref{thm:lower-two-bits-one-price}                                   & \cellcolor[HTML]{FD6864} \\ \hline
        \multicolumn{1}{|l|}{Two prices}   & \cellcolor[HTML]{32CB00} ${\Omega(\sqrt T)}$ - Thm. \ref{thm:lower-full}                 &                  $\cellcolor[HTML]{FFE135} 
 {\Omega(T^{\nicefrac 23})}$ - Thm. \ref{thm:lower-two-bits-two-prices}                                     & $\cellcolor[HTML]{FFC702} {O(T^{\nicefrac 34})}$ - Thm. \ref{thm:upper-one-bit-two-prices}  \\ \hline
        \end{tabular}
        \caption{Summary of 2-regret results in various settings.}
        \label{table:results}
    \end{table}
    
\subsection{Overview of our Results}
    
    We present our results for the adversarial sequential bilateral trade problem (see also Table \ref{table:results}, same colors correspond to same minimax regret rate).
    
\begin{itemize}
    \item We show that no learning algorithm can achieve sublinear $\alpha$-regret for any $\alpha<2$ (\Cref{thm:lower-full-2-eps}). This holds in the full feedback model (and thus for both partial feedback models). 
    \item We give a \new{single-price} learning algorithm with full feedback that achieves $\tilde O(\sqrt{T})$ $2$-regret\footnote{The $\tilde O$ hides poly-logarithmic terms} (\Cref{thm:upper-full}) and show that no algorithm can improve upon this (\Cref{thm:lower-full}). \new{Notably, our lower bound holds even for algorithms that are allowed to post two prices.}
    \item We show that if limited to a single price, no learning algorithm achieves sublinear $\alpha$-regret for any constant $\alpha$ in either partial feedback models, {\sl i.e.}. one or two-bit feedback (\Cref{thm:lower-two-bits-one-price}). 
    \item Given the negative results above, we show that allowing the learning algorithm to post two prices gives sublinear \new{(in particular, $O(T^{\nicefrac 34})$)} $2$-regret even for one-bit feedback (\Cref{thm:upper-one-bit-two-prices}). This means that our learning algorithm achieves, on average, at least half of the gain from trade of the best fixed price in hindsight, using only one-bit of feedback at each step\new{.}
    \item We show a separation between partial versus full feedback by giving a $\Omega(T^{\nicefrac 23})$ lower bound on the $2$-regret achievable by any learning algorithm in the two prices and two-bit feedback model (\Cref{thm:lower-two-bits-two-prices}\footnote{The lower bound construction given in the conference version of this paper \citep{AzarFF22} does not provide the required bound and hence the version presented here is the appropriate one.}).
\end{itemize}    
 
    The gaps in Table 1 may appear misleading because upper bounds in weaker models apply in stronger models and lower bounds in stronger models apply in weaker models. The only remaining open gap in our results is between the $\Omega(T^{\nicefrac 23})$ lower bound (light orange in \Cref{table:results}) and the $O(T^{\nicefrac 34})$ upper bound (dark orange in \Cref{table:results}) that hold for two prices and partial feedback (either one or two-bit feedback). It is also worth noting that in our worst case model two prices are required but one-bit suffices for sublinear $2$-regret. This is a different qualitative behaviour that the one observed in the stochastic case \citep{Nicolo21}, where it is enough to use one single price but the two-bit feedback is required to achieve sublinear ($1$-)regret. One may wonder why two prices are helpful at all in our adversarial setting, given their suboptimality in maximizing the gain from trade. It turns out that \new{by} randomizing over two prices it is possible to estimate the (non-stochastic) valuations of the agents.
    % As already mentioned before, choosing a single price between the posted price to the seller and the posted price of the buyer is always at least as successful in terms of the success of the bilateral trade. 
    % This is a bit surprising as the sequence is worst case and there is no correlation between the history and future valuations of pairs of buyer and seller. In particular, what we learned at step $t$ may be, in principle, completely irrelevant for future requests. 
    % It is also worth noting that in our worst case model two prices are required but one-bit suffices for $2$-regret. This is a different qualitative behaviour that the one observed in the stochastic case \citep{Nicolo21}, where it is enough to use one single price but two-bit feedback is required to achieve sublinear ($1$-)regret. 

\subsection{Technical challenges}
 
    \paragraph{From experts to prices} As already observed in \citet{Nicolo21}, the full feedback model nicely fits into the prediction with experts framework \citep{Nicolo06}: there is a clear mapping between expert\new{s} and prices and the mechanism can easily reconstruct the gain that each price/expert experiences using the feedback received. The main challenge here is given by the continuous nature of the possible prices, as the usual experts framework assumes a finite number of experts. There are workarounds that exploit some regularity of the gain function such as the Lipschitz property or convexity/concavity \cite[see, e.g.,][]{Nicolo06,Hazan16,Slivkins19}. Unfortunately, gain from trade is not such a function. Moreover, in our adversarial setting we cannot adopt the smoothing trick used in \citet{Nicolo21}, where they assume some regularity on the agents distribution to argue that $\E{\gft(\cdot)}$ becomes Lipschitz. Our main technical tool to address this issue is a discretization claim that allows us to compare the performance of the best fixed price in $[0,1]$ with that of the best on a finite grid. 
    
    \paragraph{A magic estimator} Consider the two partial feedback models; there at each time step $t$ the learner only receives a minimal information about what happened at time $t$: namely, one or two-bit versus the full knowledge of $\gft_t(\cdot)$. This type of feedback is strictly more difficult than the classic bandit feedback \citep{Nicolo06}, where the learner always observes at least the gain its action incurred. Our main technical tool to circumvent this issue is given by the design of a procedure that, posting two randomized prices, is able to estimate the $\gft_t$ in a given price. This unbiased estimator is then used in a carefully designed block decomposition of the time horizon to achieve sublinear $2$-regret in presence of this very poor feedback. 
    
    \paragraph{Lower bounds} For our lower bounds we adopt two different strategies. In Theorems \ref{thm:lower-full-2-eps} and \ref{thm:lower-two-bits-one-price} we construct randomized instances where no algorithm can learn anything: the only prices the learner could use to discriminate between different instances are cautiously hidden, while all the other prices do not reveal any useful information, given the type of feedback considered. The randomized instances used in Theorems \ref{thm:lower-full} and \ref{thm:lower-two-bits-two-prices} involve instead a more structured approach; this is due to the challenge posed by the contemporary handling of the multiplicative and additive part of the $2$-regret. To this end, we carefully hide the optimal ex-post prices and make hard to the learner to achieve small ($1$-)regret with respect to some ``second best'' prices.
    A crucial task we often face is to ``hide" some small finite sets of critical prices from the learning algorithm. We employ two techniques to do so: repeatedly dividing overlaps (Theorems \ref{thm:lower-full-2-eps} and \ref{thm:lower-full}) and random shifts (as in the proof of Theorems \ref{thm:lower-two-bits-one-price} and \ref{thm:lower-two-bits-two-prices}).

\subsection{Related work}

    The work that is most closely related to ours is \citet{Nicolo21}. There, the authors study the same sequential bilateral trade problem as we do, with the objective of minimizing the ($1$-)regret with respect to the best fixed price. They focus on the (easier) stochastic model, where the adversary chooses a distribution over valuations and not a deterministic sequence like in our model. A full characterization of the minimax regret regimes is offered, for the same type of feedback we consider (note that the one-bit feedback is only addressed in their extended version \citep{Nicolo21long}) and with various regularity assumptions on the underlying random distributions. \citet{Nicolo21} also give the first result for the adversarial setting we consider, showing that no learning algorithm can achieve sublinear $1$-regret.
    
    \paragraph{Online Learning and Economics} Regret minimization has been applied to many economic models: auctions \citep{morgenstern2015pseudo,Cesa-BianchiGM15,LykourisST16,DaskalakisS16,WeedPR16,BalseiroG19}, contract design \citep{HoSV16,ZhuBYWJJ22,DuettingGSW23}, and Bayesian persuasion \citep{CastiglioniCMG23}. In particular, \citet{kleinberg2003value} studied the one-sided pricing problem, offering a $\tilde O(T^{\nicefrac 23})$ upper bound on the regret in the adversarial setting and opening a fruitful line of research \citep{blum2004online,blum2005near,bubeck2017online}.
    
    \paragraph{The notion of $\alpha$-regret} The notion of $\alpha$-regret has been formally introduced by \citet{KakadeKL09}, but was already present in \citet{KalaiV05}. It has then found applications in linear \citep{Garber21} and submodular optimization \citep{RoughgardenW18}, learning with sleeping actions \citep{Emamjomeh-Zadeh21}, combinatorial auctions \citep{RoughgardenW19} and market design \citep{NiazadehGWSB21}. We mention that the notion of $\alpha$-regret is commonly referred to as competitive ratio in the literature concerning adversarial bandits with knapsack \citep[see, e.g.,][]{ImmorlicaSSS22}. 
    
    \paragraph{Online learning with partial feedback} Our work fits in the line of research that studies online learning with feedback models different from full information and the bandit ones; our one and two-bit feedback models share similarities with the feedback graphs model \citep[see, e.g., ][]{AlonCGMMS17,HoevenFC21,EspositoFHC22} and the partial monitoring framework \citep[see, e.g., ][]{BartokFPRS14,LattimoreS19}.   
    
     \paragraph{Bilateral Trade} While \citet{Myerson83} were the first to thoroughly investigate the bilateral trade problem in the Bayesian setting with their famous impossibility result, it was  the seminal paper of Vickrey \citep{Vickrey61} that introduced the problem, proving that any mechanism that is welfare maximizing, individually rational, and incentive compatible may not be budget balanced. In the Bayesian setting, it was only very recently that \citet{DengMSW21} gave the first (Bayesian) incentive compatible, individually rational and budget balanced mechanism achieving a constant factor approximation of the optimal gain from trade. Prior to this paper, a  posted price $O\big(\log\frac 1r\big)$ approximation bound was achieved by \citet{Colini-Baldeschi17}, with $r$ being the probability that a trade happens (i.e., the value of the buyer is higher than the value of the seller). The literature also includes many individually rational, incentive compatible and budget balanced mechanisms achieving a constant factor approximation of the optimal social welfare. \citet{BlumrosenD14} proposed a simple posted price mechanism, the median mechanism, yielding a $2$-approximation of the optimal social welfare; the same authors then implemented a randomized fixed price mechanism improving the approximation to $\nicefrac e{e-1}$ \citep{BlumrosenD21}. Very recently, \citet{CaiW23} and \citet{LiuR023} further improved the state of the art for welfare maximizing fixed prices mechanisms, by providing a $0.71$ approximation. A related line of work investigates the single-sample version of the problem. Recently, \citet{Duetting20} showed that even posting one single sample from the seller distribution as price is enough to achieve a $2$ approximation to the optimal social welfare, this result has been further improved by \citet{CaiW23} and \citet{LiuR023}. The class of fixed price mechanism is of particular interest as it has been showed that all (dominant strategy) incentive compatible and individually rational mechanisms that enforce a stricter notion of budget balance, i.e., the so-called strong budget balance (where the mechanism is not allowed to subsidize or extract revenue from the agents) are indeed fixed price \citep{hagerty1987robust,Colini-Baldeschi16}.

    \paragraph{Conference Version and follow-up Work} A preliminary version of this work appeared in \citet{AzarFF22}, while a journal version appeared in \citet{AzarFF24}. In a follow-up work, \citet{Cesa-BianchiCCF23} studied the repeated bilateral trade problem against a smooth adversary, where the sequence of valuation is generated following smooth (non-stationary) distributions generated by an adversary. The authors provide tight results for this smoothed version of the problem both in the full and partial feedback. Recently, \citet{BernasconiCCF24} provided the first no-regret algorithms in the adversarial setting by moving from a per-step notion of budget balance (like the one enforced here) to a (weaker) global one. 

\section{Preliminaries}\label{sec:preliminaries}

\begin{algorithm*}[t!]
\caption*{\textbf{Learning Protocol of Sequential Bilateral Trade}}
\label{a:learning-model}
\begin{algorithmic}
\For{time $t=1,2,\ldots$}
    \State a new seller/buyer pair arrives with (hidden) valuations $(s_t,b_t) \in [0,1]^2$
    \State the learner posts prices $p_t, q_t \in [0,1]$
    \State the learner receives a (hidden) reward $ \gft_t(p_t, q_t) \in [0,1]$
    \State a feedback $z_t$ is revealed
\EndFor
\end{algorithmic}
\end{algorithm*}

The formal protocol for the sequential bilateral trade follows \citet{Nicolo21}. At each time step $t$, a new pair of seller and buyer arrives, each with private valuations $s_t$ and $b_t$ in $[0,1]$; the learner posts two prices: $p_t \in [0,1]$ to the seller and $q_t\in [0,1]$ to the buyer. A trade happens if and only if both agents agree to trade, i.e., when $s_t \le p_t$ and $ q_t \le b_t.$
Since we want our mechanism to enforce budget balance, we require that $p_t \le q_t$ for all $t.$ When a trade occurs, the learner is awarded with the resulting increase in social welfare, i.e., $b_t - s_t.$ The learner then observes some feedback $z_t.$ The gain from trade at time $t$ depends on the valuations $s_t$ and $b_t$ and on the price posted. To simplify the notation we introduce the following:
\[
    \gft_t(p, q) := \gft(p, q, s_t,b_t) =\ind{s_t \le p \le q \le b_t}\cdot (b_t - s_t)
\]
When the two prices are equal, we omit one of the arguments to simplify the notation.

Given any constant $\alpha \ge 1$, the $\alpha$-regret of a learning algorithm $\A$ against a sequence of valuations $\S$ on time horizon $T$ is defined as follows
\[
     R_T^\alpha(\A,\S) :=  \max_{p,q \in [0,1]^2}\sum_{t=1}^T \gft_t(p,q) - \alpha \cdot \sum_{t=1}^T \E{\gft_t(p_t,q_t)}.
\]
In the right side of the equality the dependence on $\S$ is contained in the $\gft_t(\cdot)$. Note that the expectation in the previous formula is with respect to the internal randomization of the learning algorithm: $p_t$ and $q_t$ are the (possibly random) prices posted by $\A$.

The $\alpha$-regret of a learning algorithm $\A$, without specifying the dependence of the sequence, is defined as its $\alpha$-regret against the ``worst'' sequence of valuations: $
    R_T^\alpha(\A) := \sup_{\S} R_T^\alpha(\A,\S)$. 
Stated differently, the performance of an algorithm is measured against an oblivious adversary that generates the sequence of valuations ahead of time: the learner has to perform well on {\em all} possible sequences. In this paper we study the minimax $\alpha$-regret, $R^{\alpha, \star}_T$, that measures the performance of the best (learning) algorithm versus the optimal fixed price in hindsight, on the worst possible instance: $R^{\alpha, \star}_T := \inf_{\A}R_T^\alpha(\A)$.
The set of learning algorithms we consider depends on which of the various settings we are dealing with. In this paper we consider a variety of such settings ({\sl i.e.}, how many prices are posted, what feedback is available, see below Sections \ref{subsection:prices} and \ref{subsection:feedback}). 

\subsection{Single Price {\sl vs.} Two Prices --- Seller price and Buyer price} \label{subsection:prices}
    
    We consider two families of learning algorithms, differing in the nature of the probe they perform, corresponding to two notions of what it means to be budget balanced: 
    
    \paragraph{Single price} If we want to enforce a stricter notion of budget balance, namely strong budget balance, the mechanism is neither allowed to subsidize nor extract revenue from the system. This is modeled by imposing $p_t = q_t$, for all $t$.  If $p_t = q_t$ we use the notation $\gft_t(p_t)$ to represent the gain from trade at time $t$.  
    
    \paragraph{Two prices} If we require that the mechanism enforces (weak) budget balance, it can post two different prices, $p_t$ to the seller and $q_t$ to the buyer, as long as $p_t\le q_t$. {\sl I.e.}, we only require that the mechanism never subsidize a trade, we do not require that the mechanism not make a profit. In this setting we use the notation $\gft_t(p_t,q_t)$ to represent the gain from trade at time $t$.   
    
    \begin{observation} \label{obs:twopriceobs}
        Note that the only reason to post two prices is to obtain information. For any pair of prices $(p,q)$ with $p < q$ posting any single price $\pi \in [p,q]$ guarantees no less gain from trade.
    \end{observation}
    
    In particular, any budget balanced algorithm that knows the future and seeks to maximize gain from trade while repeatedly posting the same prices will never choose two different prices.

\subsection{Feedback models} \label{subsection:feedback}

    We consider three types of feedback, presented here in increasing order of difficulty for the learner. (Note that   full feedback ``implies'' two-bit feedback which in turn implies one-bit feedback):
    
    \paragraph{Full feedback} In the full feedback model, the learner receives both seller and buyer valuations, immediately after posting prices the feedback to the learner at time $t$: formally, $z_t = (s_t,b_t)$. E.g., both seller and buyer send sealed bids that are opened immediately after the (one or two) price(s) are revealed. It follows from Observation \ref{obs:twopriceobs} that in the full feedback model there is never any reason to post two prices, as all the relevant information is revealed anyway. 
    
    \paragraph{Two-bit feedback}  In two-bit feedback the algorithm observes separately if the two agents agree on the given price, {\sl i.e.,} the feedback at time $t$ is $z_t = (\ind{s_t \le p_t}, \ind{q_t \le b_t})$.
    
    \paragraph{One-bit feedback} The one-bit feedback is arguably the minimal feedback the learner could get: the only information revealed is whether the trade occurred or not, {\sl i.e.}, $z_t = \ind{s_t \le p_t \le q_t \le b_t}$. 
    
    \subsection{Lower bounds via Yao's Minimax Theorem}
    
        An important technical tool we use to prove our lower bounds is the well known Yao's Minimax Theorem \citep{Yao77}. In particular, we apply the easy direction of the theorem, which reads (using our terminology) as follows: the $\alpha$-regret of a randomized learner against the worst-case valuations sequence is at least the minimax \new{$\alpha$-}regret of the optimal deterministic learner against a stochastic sequence of valuations. Formally,
        \[
            R^{\alpha,\star}_T \ge  \sup_{\A} \E{\max_{p,q \in [0,1]^2}\sum_{t=1}^T \gft_t(p,q) - \alpha \cdot \sum_{t=1}^T \gft_t(p_t,q_t)},
        \]
        where the expectation is with respect to the stochastic valuation sequence $\S$, while $\A$ denotes a deterministic learner. We remark that --- for the minimax theorem to be applicable --- the random instance $\S$ has to be oblivious of the learner.

\subsection{Regret due to discretization}

    Our first theoretical result concerns the study of how discretization impacts the regret. In particular, we compare the performance of the best fixed price taken from the continuous set $[0,1]$ to that of the best fixed price chosen from some discrete grid $Q\subset [0,1]$. Optimizing over a continuous set may seemingly be a problem because our object, gain from trade, is discontinuous (thus non-Lipschitz), non-convex and non-concave; one cannot use the ``standard approach" that makes use of such regularity conditions. What we show in the following \new{c}laim is that it is possible to compare the performance of the best continuous fixed price with {\em twice} that of the best fixed price on the grid. 
\begin{claim}[Discretization error]
    \label{claim:discretization}
        Let $Q = \{ q_0 = 0\leq q_1\leq  q_2 \cdots\leq  q_n = 1\}$ be any finite grid of prices in $[0,1]$ and let $\delta(Q)$ be the largest difference between two contiguous prices, {\sl i.e.}, $\max_{i=1,\ldots,n} |q_i - q_{i-1}|$,
        then for any sequence $\S = (s_1,b_1), \ldots, (s_T,b_T)$ and any price $p$ we have
        \[
        % \begin{equation}
        % \label{eq:discretization}
            \sum_{t=1}^T \gft_t(p) \le 2 \cdot \max_{q \in Q} \sum_{t=1}^T \gft_t(q) + \delta(Q) \cdot T.         
        % \end{equation}
        \]
        Let $\A$ be any learning algorithm that posts prices $(p_t,q_t)$, then the following inequality holds:
        \begin{equation}
        \label{eq:discretization_regret}
            R^2_T(\A) \le 2\sup_{\S} \left\{\max_{q \in Q} \sum_{t=1}^T \gft_t(q) - \sum_{t=1}^T \E{\gft_t(p_t,q_t)} \right\}+ \delta(Q)\cdot T.
        \end{equation}
    \end{claim}

    \begin{proof}
        Fix any sequence of valuations $\S$ and let $p^{\star}$ be the best fixed price on for that sequence: $
            p^{\star} \in \argmax_{p \in [0,1]} \sum_{t=1}^T \gft_t(p).$
        If $p^{\star} \in Q$, then there is nothing to prove; alternatively let $\qplus$ and $\qminus$ be the consecutive prices on the grid such that $p^{\star} \in [\qminus,\qplus]$. For any time step $t$ such that $\gft_t(p^{\star})>0$, either $p^{\star} \in [s_t,b_t] \subseteq [\qminus,\qplus]$, in which case 
        \[
        \gft_t(p^{\star}) \new{=} (b_t - s_t) \le (\qplus-\qminus) \le \delta(Q),
        \]
        or $[s_t,b_t] \cap \{\qplus, \qminus\} \neq \emptyset$, so that $\gft_t(p^{\star}) = \max\{\gft_t(\qplus), \gft_t(\qminus)\}$.
        All in all, we have that, for each time $t$, the following inequality holds:
        \[
            GFT_t(p^{\star}) \le GFT_t(\qplus) + GFT_t(\qminus) + \delta(Q)
        \]
        Summing up over all times $t$ we get:
        \begin{align}
        \nonumber
            \sum_{t=1}^T \gft_t(p^{\star}) &\le \sum_{t=1}^T GFT_t(\qplus) + \sum_{t=1}^T GFT_t(\qminus) + \delta(Q) T\\
            \label{eq:discretization}
            &\le 2 \cdot \max_{q \in Q} \sum_{t=1}^T \gft_t(q) + \delta(Q)T.
        \end{align}
        Focus now on the second part of the claim and fix any learning algorithm $\A$, we have:
        \begin{align*}
            R^2_T(\A) =& \sup_{\S}\left\{\sum_{t=1}^T GFT_t(p^{\star}) - 2 \cdot \sum_{t=1}^T \E{GFT_t(p_t,q_t)}\right\}\\
            \le& \sup_{\S}\left\{\sum_{t=1}^T GFT_t(p^{\star}) - 2 \cdot \max_{q \in Q} \sum_{t=1}^T GFT_t(q) \right\} \\
            &+ \sup_{\S}\left\{2 \max_{q \in Q} \sum_{t=1}^T GFT_t(q) - 2  \sum_{t=1}^T \E{GFT_t(p_t,q_t)}\right\}\\
            \le& T \cdot  \delta(Q) + 2 \sup_{\S}\left\{\max_{q \in Q} \sum_{t=1}^T GFT_t(q) - \sum_{t=1}^T \E{GFT_t(p_t,q_t)}\right\},
        \end{align*}
        where the last \new{passage} follows from \new{Inequality~\ref{eq:discretization}} (that holds for all sequences).
    \end{proof}

    Before moving to the next section, we spend some words to compare our discretization result with the one in \citet{Nicolo21} (Second decomposition Lemma). There the authors exploit the stochastic nature of the valuations to argue that $\E{\gft_t(\cdot)}$ is Lipschitz, under some regularity assumptions on the random distributions. We study the adversarial model, thus we cannot use this ``smoothing'' procedure; this is why we lose an extra multiplicative factor of $2$.

\section{Full Feedback}\label{sec:full}

    In this section we study the full feedback model, where the learner receives as feedback both seller and buyer valuations after posting a single price (see Observation \ref{obs:twopriceobs}). The learner can thus evaluate $\gft_t(p)$ for all $p \in [0,1]$, independently by the price posted. Even with this very rich feedback we show that the impossibility result from \citet{Nicolo21}, {\sl i.e.}, no learning algorithm achieves sublinear ($1$-)regret in the sequential bilateral trade problem, can be extended to hold for  $\alpha$-regret for all $\alpha \in [1,2)$. 
    
    To prove this result, formalized in Theorem \ref{thm:lower-full-2-eps}, we use Yao's Minimax Theorem: a randomized family of valuations sequences is constructed, with the property that any deterministic learner would suffer, on average, linear $2-\eps$ regret against it. The detailed proof is provided below, but we sketch here the main ideas. Specifically, any valuations sequence from the randomized family consists of $(sell,buy)$ prices that have the form $(0,b_i)$ or $(s_i,1)$, for some carefully designed sequences $\{s_i\}_i$ and $\{b_i\}_i$. These sequences are generated iteratively in a way such that all realized $[sell,buy]$ segments overlap and the next segment is chosen at random among two disjoint options $(0,b_i)$ or $(s_i,1)$. Since all realized $[sell,buy]$ segments overlap, there is at least one price in the intersection of all intervals: this is the optimal fixed price in hindsight. Conversely, at each time step no learner can post a price that guarantees a trade with probability greater than $\nicefrac 12$, thus yielding the lower bound.

\begin{theorem}[Lower bound on $(2-\eps)$-regret]
\label{thm:lower-full-2-eps}
    In the full-feedback model, for all $\eps \in (0,1]$ and horizons $T$, the minimax $(2-\eps)$-regret satisfies $
        R_T^{2-\eps,\star} 
        \ge \frac 18 \eps T.$
\end{theorem}

\begin{figure}
    \centering
    \includegraphics[width=0.6\linewidth]{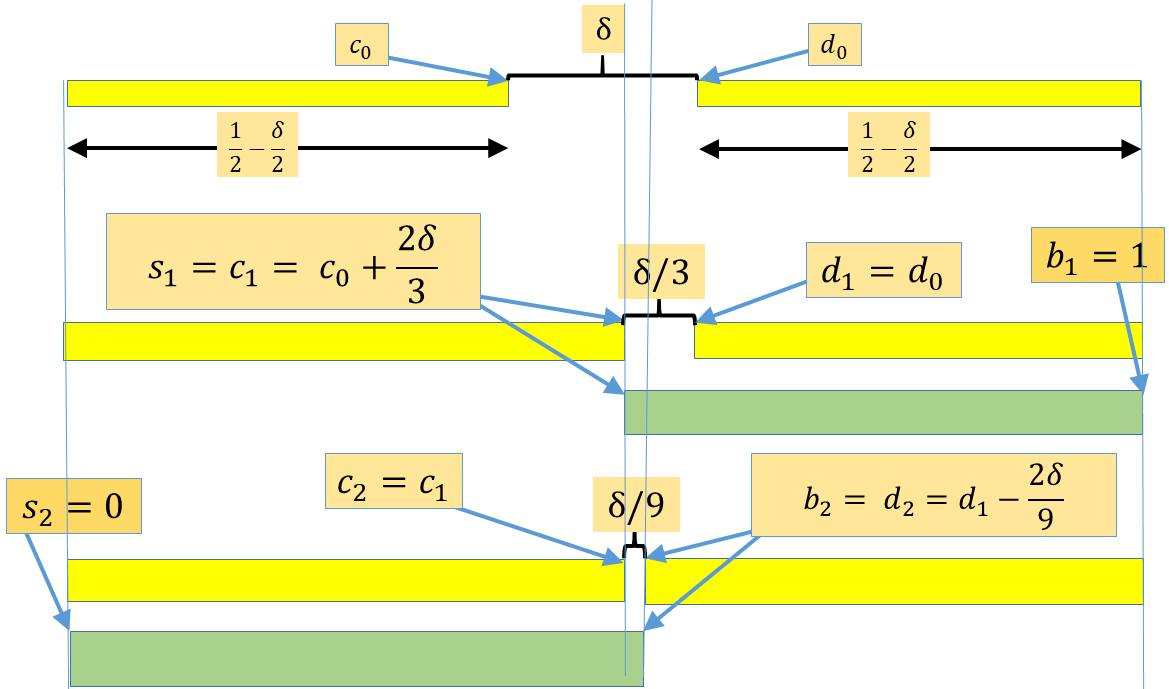}
    \caption{Lower bound construction that ``hides" the optimal price.}
    \label{fig:regret1}
\end{figure}

\begin{proof}
    We prove this lower bound via Yao's Theorem. Fix any $\eps \in (0,1]$, we argue that there exist some constant $c_{\eps}$ and a distribution over sequences such that the $(2-\eps)$-regret of any deterministic learning algorithm $\A$ against it is, on average, at least $c_{\eps}\cdot T,$ for some $\eps$-dependent constant $c_{\eps}$. Our construction is reminescent ---and to some extent simplifies--- the one given in Theorem 4.6 of \citet{Nicolo21}, but presents one main difference: here we construct a family of instances that is oblivious to the learner, whereas in \citet{Nicolo21} they construct a single instance tailored to the learning algorithm $\A$. It is critical for the application of Yao's Theorem that the sequence distribution be independent of the actual algorithm. 
    
    Let $\delta < \nicefrac \eps8 $, the adversary initiates two auxiliary sequences of points $c_0 = \nicefrac 12 - \nicefrac \delta 2$ and $d_0 = \nicefrac 12 + \nicefrac \delta 2$ then, inductively constructs the auxiliary sequences and draws $s_{t+1}$ and $b_{t+1}$ as follows:
    \[
        \begin{cases}
            c_{t+1}:=c_{t}, \ d_{t+1}:=d_{t}-\frac{2\delta}{3^{t}}, \ s_{t+1}:=0, \ b_{t+1}:=d_{t+1},
        &
             \text{ with probability } \nicefrac 12\;
        \\
            c_{t+1}:=c_{t}+\frac{2\delta}{3^{t}}, \ d_{t+1}:=d_{t}, \ s_{t+1} := c_{t+1}, \ b_{t+1}:=1,
        &
            \text{ with probability } \nicefrac 12.
        \end{cases}
    \]
    A quick description of the procedure (we refer to \Cref{fig:regret1} for a pictorial representation): at the beginning of each time step $t+1$ the adversary has two points, $c_t$ and $d_t$, with $d_t - c_t = 3^{-t}\delta $. Then, it chooses uniformly at random between \emph{left} or \emph{right}.
    If \emph{left} is chosen (first line of the construction and $t=2$ in \Cref{fig:regret1}), then $c_{t+1} = c_t$, while $d_{t-1}$ is moved to the first third of the $[c_t,d_t]$ interval: $
        d_{t+1} = d_{t} - \nicefrac{2}{3^t} \delta = c_t + \nicefrac{1}{3^t}\delta,$
    and the adversary posts prices $s_{t+1} = 0$ and $b_{t+1} = d_{t+1}$. If \emph{right} is chosen (second line of the construction and $t=1$ in \Cref{fig:regret1}), then the symmetric event happens: $d_{t+1}=d_t$, $c_{t+1}$ moves to the second third of $[c_t,d_t]$ and the adversary posts $s_{t+1} = c_{t+1}$ and $b_{t+1} = 1$. 
    
    At each time step the two possible realizations (for a fixed past) of the $[s_{t}, b_{t}]$ intervals are disjoint: it implies that {\em any price} the learner posts results in a trade with probability (over the randomness of the adversary) of at most $\nicefrac 12$.
    As we are in a full feedback scenario, there is no point for the learner to post two prices, so we assume that $\A$ posts a single price. By construction, we also have that $(b_t - s_t) \le \nicefrac{1}2 + \nicefrac\delta 2 $ for all realizations. Thus the total expected gain from trade achieved by any deterministic learning algorithm is:
    \begin{align}
    \label{eq:bound_alg}
        \E{\sum_{t=1}^T \gft_t(p_t)} \le \sum_{t=1}^T \frac{1}{2} \left( \frac{1}2 + \frac\delta 2\right) = \frac{T}{2}\left( \frac{1}2 + \frac\delta 2\right)
    \end{align}

    We move now our attention to the best fixed price in hindsight. For any realization of the randomness used in the construction of the sequence, $[s_{t}, b_{t}]$  intervals have a non-empty intersection; let $p^{\star}$ be some price in this intersection. Moreover, at all time steps $t$ it holds that $(b_t-s_t) \ge \nicefrac{1}2 - \nicefrac\delta 2 $. All in all, this gives a simple lower bound on the total gain from trade of the best price in hindsight that holds for any realization of the valuations sequence:
    \begin{equation}
        \label{eq:bound_opt}
        \max_{p \in [0,1]} \sum_{t=1}^T \gft_t(p) = \sum_{t=1}^T \gft_t(p^{\star}) \ge \frac T2  \left(1-\delta \right).
    \end{equation}
    Combining these two bounds (\new{Inequalities~\ref{eq:bound_alg} and \ref{eq:bound_opt}}) we can derive the desired bound on the $2-\eps$ regret via Yao's Theorem:
    \[
        R^{2-\eps}_T(\A) \ge \frac T2  \left(1-\delta \right) - (2-\eps)\frac{T}{2}\left( \frac{1}2 + \frac\delta 2\right) = \frac T4 (\eps + \eps \delta - 4\delta) \ge \frac 18 \eps T. \qedhere
     \]
\end{proof}

If we look for positive results, we note that there is a clear connection of our problem in the full feedback and the prediction with experts framework \citep{Nicolo06}. In particular, if we simplify the task of the learner and ask it to be competitive against {\em the best price in a finite grid}, we can use classical results on prediction with experts as a black box. Combining this fact with our discretization result (\Cref{claim:discretization}), we can show an $\tilde O(\sqrt{T})$ upper bound on the $2$-regret.
    \begin{theorem}[Upper bound on $2$-regret with full feedback]
    \label{thm:upper-full}
        In the full-feedback setting, there exists a learning algorithm $\A$ whose $2$-regret, for $T$ large enough, respects $
            R_T^2(\A) \le 5 \cdot \sqrt{T \log T}.$
    \end{theorem}

    \begin{proof}
    Consider a grid of prices $Q$ composed by $T+1$ equally spaced points: $q_i = \nicefrac iT$ for $i = 0, 1, \dots, T$ and choose your favourite prediction with experts learning algorithm, e.g., Multiplicative Weights \citep{AroraHK12}. Given the full feedback regime, and the fact that the grid is finite, we can run expert algorithm using as experts the points on the grid. Typically, the best experts learning algorithm exhibit a bound on the regret $O(\sqrt{T \log K})$, that becomes $O(\sqrt{T \log T})$ in our case since we have $T+1$ experts. If we use the Multiplicative Weights algorithm against the best fixed price on the grid $Q$ with $\eta = \sqrt{\nicefrac{\log T}{T}}$, we get by Theorem 2.5 of \citet{AroraHK12}:
    \[
        \sup_\S\left\{\max_{q \in Q} \sum_{t=1}^T \gft_t(q) - \sum_{t=1}^T \E{\gft_t(p_t)} \right\} \le 2 \sqrt{T \log(T+1)}
    \]
    Plugging this bound in \Cref{claim:discretization}, we get the desired order of regret. 
    \begin{align*}
         R^2_T(\A) &\le 2\sup_\S\left\{\max_{q \in Q} \sum_{t=1}^T \gft_t(q) - \sum_{t=1}^T \E{\gft_t(p_t)} \right\}+ \delta(Q)\cdot T\\
         &\le 4 \sqrt{T \log(T+1)} + 1 \le 5 \sqrt{T \log(T)}.
    \end{align*}
    The first inequality is just a restatement of \new{Inequality~\ref{eq:discretization}} from \Cref{claim:discretization}. The second inequality follows by combining the bound on the regret of multiplicative weight and the fact that the grid is equally spaced, thus $\delta(Q) = \nicefrac 1T.$
\end{proof}

    We conclude the analysis of repeated bilateral trade in the full feedback model with a lower bound that shows that the previous result is tight up to a logarithmic factor: the minimax $2$-regret of the full feedback problem is $\tilde \Theta(\sqrt{T})$.  The proof uses once again Yao's Theorem and consists in constructing a randomized family of sequences such that any deterministic learning algorithm suffers, in expectation, a $\Omega(\sqrt{T})$ $2$-regret. The detailed construction is described below and it involves the careful combination of two scaled copies of the hard sequences used in the proof of \Cref{thm:lower-full-2-eps}. As a technical ingredient, we need the following property of Random walks. 
    \begin{lemma}[Property of Random Walks]
    \label{lem:random}
        Let $S_n$ be a symmetric random walk on the line after $n$ steps\footnote{\new{Formally, the random walk is recursively defined as follows: $S_0=0$, and $S_{n}$ is either  $S_{n-1} + 1$ or $S_{n-1} - 1$, with the same probability.}}. Then, for $n$ large enough, it holds that $
          \E{|S_n|} \ge \frac 23 \sqrt{n}.$
    \end{lemma}
    \begin{proof}
    It is well known that the expected distance of a random walk from the origin grows like $\Theta(\sqrt{n})$. Formally, the following asymptotic result holds \citep[e.g.,][]{Palacios08}
        \[
            \lim_{n \to \infty}\tfrac{\E{|S_n|}}{\sqrt n} = \sqrt{\tfrac{2}{\pi}}.
        \]
    Observe that $\sqrt{\nicefrac{2}{\pi}} > \nicefrac 23$, thus there exists a finite $n_0$ such that $\E{|S_n|} \ge \nicefrac 23 \sqrt{n}$ for all $n \ge n_0.$
    \end{proof}

    \begin{theorem}[Lower bound on $2$-regret with full feedback]
    \label{thm:lower-full}
        In the full-feedback model, for all horizons $T$ large enough, the minimax $2$-regret satisfies $ R_T^{2,\star} \ge \frac 1{13}\sqrt{T}. $
    \end{theorem}

  \begin{figure}
    \centering
    \includegraphics[scale=0.35]{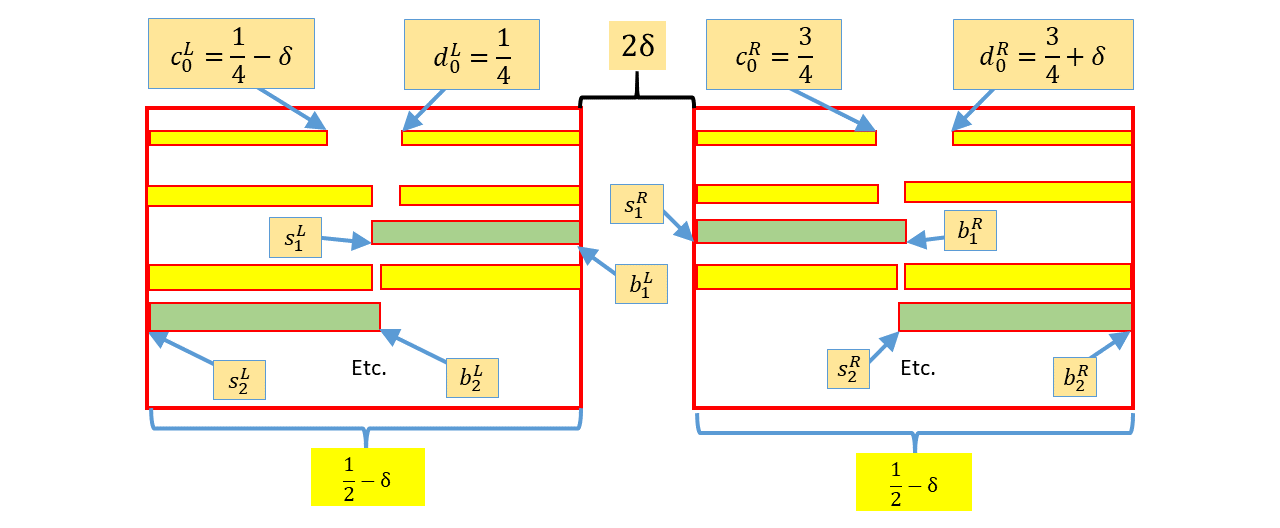}
    \caption{The proof of Theorem \ref{thm:lower-full} makes use of two (appropriately scaled and shifted) copies of the lower bound from Theorem \ref{thm:lower-full-2-eps} (See Figure \ref{fig:regret1}).  
    % Note that the random coins used in the construction of the $(s^L,b^L)$ and $(s^R,b^R)$ are independent. 
    In this example the left hand copy choose right and then left, while the right hand copy happened to choose left and then right. The (seller,buyer) bids at time $t$ are then chosen independently at random from $(s_t^L,b_t^L)$ and $(s_t^R,b_t^R)$.
    }
%    \vspace{-10pt}
    \label{fig:2copies}
\end{figure}

\begin{proof}
    We show that there exists a distribution over valuations sequences such that any deterministic learning algorithm $\A$ achieves, on average, at least a $2$-regret of $\nicefrac 1{13}\cdot \sqrt{T}.$ This is enough to conclude the proof via Yao's Theorem. It may be helpful to consider Figure \ref{fig:2copies} to visualize this construction. Fix some small $\delta $ to be set later and consider two scaled copies of the lower bound construction from \Cref{thm:lower-full-2-eps}, one in $[0,\nicefrac 12-\delta]$ and the other is $[\nicefrac 12+  \delta, 1]$. Starting from $
        c_0^L = \nicefrac{1}{4} - \delta, \ d_0^L = \nicefrac{1}{4},  \  c_0^R = \nicefrac{3}{4},  \ d_0^R = \nicefrac{3}{4} + \delta,$
    the left $L$ and right $R$ pair of sequences evolve over time and generate two distinct sequences of valuations: $(s^L_t, b^L_t) \subseteq [0,\nicefrac 12-\delta]$ and $(s^R_t, b^R_t) \subseteq [\nicefrac 12+ \delta, 1]$. The actual sequence of valuations presented to the learner is based on these two sequences as follows: at each time step, the adversary tosses a fair coin, if it is a head, then it selects $(s_t,b_t) = (s^L_t, b^L_t)$, otherwise $(s_t,b_t) = (s^R_t, b^R_t)$. Observe that there are two independent sources of randomness in the adversary construction: the one responsible of the generation of the auxiliary sequences and the one toss of the left-right coin. 
    We give now an upper bound on the expected performance of the learner at each time step $t$. Reasoning similarly to what we did in \Cref{thm:lower-full-2-eps}, there are four disjoint intervals of the $[0,1]$ interval where a price could cause a trade, and each one of them is the one chosen by the adversary with probability $\nicefrac 14$ ($\nicefrac 12$ given by the left-right coin and another $\nicefrac 12$, independently, by the evolution of the sequences $(s_t^L,b_t^L)$ and $(s_t^R,b_t^R)$). All in all, this implies that for any price the algorithm posts, it results in a trade with probability at most $\nicefrac 14.$ Moreover, we have the property that $(b_t - s_t) \le \nicefrac 14$ at all times and for all realizations, therefore we have proved the following upper bound on the total expected gain from trade achieved by any deterministic learning algorithm: 
    \begin{equation}
        \label{eq:upper_bound_alg}
        \sum_{t=1}^T \E{\gft_t(p_t)} \le \frac T{16}.
    \end{equation}
    We move now our attention to lower bounding the gain from trade of the best price in hindsight.
    Consider any realization of the sequence of coin tosses, we know that there exist two prices $p_L^{\star}$ and $p_R^{\star}$ such that $p_L^{\star}$ guarantees a trade in every time step where the result of the left-right coin gives left, and $p_R^{\star}$ does the same when the coin gives right. In addition, we know that $(b_t-s_t) \ge \nicefrac{1}4 - \delta.$ All in all we have that, for all realizations of the randomness,
    \[
        \gft_t(p_L^{\star}) + \gft_t(p_R^{\star}) \ge \frac{1}4 - \delta.
    \]
At this point, fix the randomness of the auxiliary sequences and focus on the the one given by the coin tosses, and call $X_t$ the indicator random variable of observing left from the coin at time $t$. 
We have:
\begin{align*}
    \mathbb E& \left[\max_{p \in [0,1]} \sum_{t=1}^T \gft_t(p)\right] \\
    &=
    \E{\max_{p \in \{p_L^{\star},p_R^{\star}\}} \sum_{t=1}^T \gft_t(p)} \\
    &\ge \left(\frac 14- \delta \right)\E{\max\left\{\sum_{t=1}^T \ind{p_L^{\star}\in [s_t,b_t]}, \sum_{t=1}^T\ind{p_R^{\star}\in [s_t,b_t]}\right\}}\\
    &= \left(\frac 14- \delta \right)\E{\max\left\{\sum_{t=1}^T X_t, T -\sum_{t=1}^T X_t\right\}}\\
    &= \left(\frac 14- \delta \right)\E{\frac T2 + \frac 12 \max\left\{2\sum_{t=1}^T X_t - T, T -2\sum_{t=1}^T X_t\right\}}\\
    &= \left(\frac 14- \delta \right)\left(\frac T2 + \frac 12 \E{|S_T|}\right)\ge \left(\frac 14- \delta \right)\left(\frac T2 + {\frac{\sqrt{T}}{3}}\right).
\end{align*}
\new{The last inequality follows from \Cref{lem:random}; consider infact the random process $S_n$ which counts the number of ``left'' outcomes, minus the number of  ``right'' outcomes during the first $n$ steps of the sequence. $S_n$ is a random walk as in the statement of \Cref{lem:random}, and can be rewritten in terms of $X$ variables as follows: 
\[
    S_n = \underbrace{\sum_{t=1}^n X_t}_{``left"} - \underbrace{\left(n - \sum_{t=1}^n X_t\right)}_{``right"} = 2 \sum_{t=1}^n X_t - n.
\]
Now, the expected distance from the origin of $S_n$ is equal to the maximum between $S_n$ and $-S_n$, so that if we look at $S_T$ we get exactly that:
\[
    |S_T| = \max\left\{ 2 \sum_{t=1}^T X_t - T, T - 2 \sum_{t=1}^T X_t\right\}.
\]}
Since the previous bound holds for any realization of the auxiliary sequences, it holds also in expectation over all the randomness. 
We can finally combine the previous inequality with \new{Inequality~\ref{eq:upper_bound_alg}} and conclude by Yao's Theorem that
\begin{align*}
    R_T^{2,\star} &\ge  \E{\max_{p \in [0,1]} \sum_{t=1}^T \gft_t(p) - 2\sum_{t=1}^T \gft_t(p_t)}
    \\
    &\ge \left(\frac 14- \delta \right)\left(\frac T2 + {\frac{\sqrt{T}}{3}}\right) - \frac T8 \\
    &\ge \frac{1}{12}\sqrt{T} - \frac \delta 2 T - \delta \frac{\sqrt{T}}{3} \ge \frac{1}{13}\sqrt{T}
\end{align*}
where in the last inequality we took $\delta$ small enough, e.g., $\delta = \nicefrac 1T.$
\end{proof}

\section{Partial Feedback}\label{sec:partial}
    
    In this section, we study the partial feedback models where the learner receives very limited information on the realizations of the gain from trade. Specifically, one or two bits that describe the relative positions of the prices proposed to the agents and their valuations.

\subsection{Lower bound on \texorpdfstring{$\alpha$}\ -regret posting single price given two-bit feedback}
    
    Consider a learner that is constrained to post one single price at every iteration; the same to both seller and buyer. For this class of algorithms we show a very strong impossibility result, namely that for any constant $\alpha$, there exists no algorithm achieving sublinear $\alpha$-regret. We prove this in the two-bit feedback model and thus it trivially holds also if given one-bit feedback. The core of the lower bound construction resides in the possibility for the adversary to hide a {\em large} interval between many shorter ones; a learner posting only one price will not be able to locate it using partial feedback (which consists in just {\em counting} the number of intervals on the left and on the right).
    
\begin{theorem}[Lower bound on $\alpha$-regret posting single price, two-bit feedback]    
\label{thm:lower-two-bits-one-price}
    In the two-bit feedback model where the learner is allowed to post one single price, for all horizons $T \in \N$ and any constant $\alpha > 1$, the minimax $\alpha$-regret satisfies $R_T^{\alpha,\star} \ge \frac{1}{128 \alpha^3} T.$
\end{theorem}

\begin{proof}
    In this proof, we construct a randomized family of sequences that are impossible to distinguish using a single price and given two-bit feedback. Furthermore, no deterministic algorithm is capable of achieving good regret against them in expectation. It may be useful to refer to Figure \ref{fig:my_label} for visualization. We first prove the claim under a {\sl ``grid hiding''} assumption that the learning algorithm is disallowed from posting prices in some fixed finite grid (to be defined below). We later  justify the grid hiding assumption by introducing some minor perturbation to the grid.   
    
\begin{figure}
    \centering
    \includegraphics[scale=0.35]{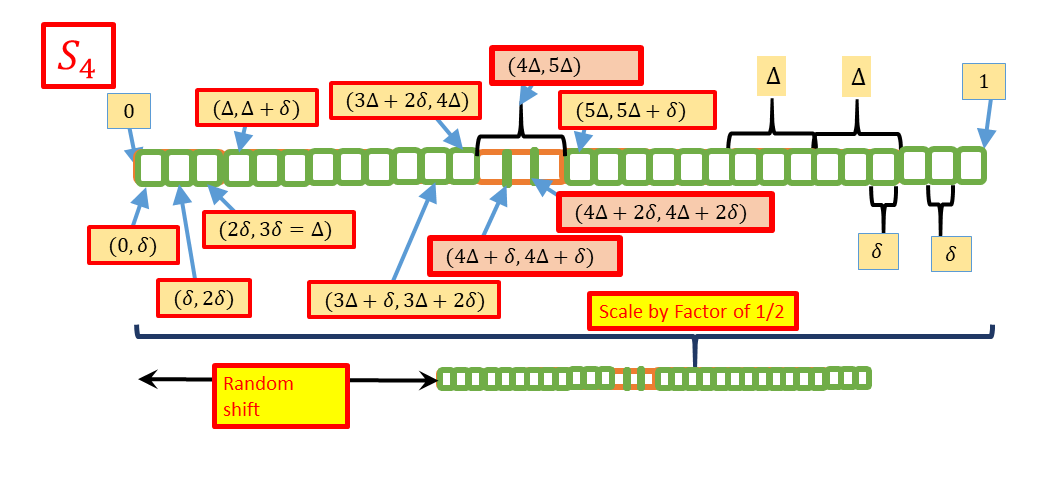}
    \caption{Example of sets used in the lower bound of  \Cref{thm:lower-two-bits-one-price} and how the grid is hidden. This example has $\Delta =\nicefrac 1{10}$, $\delta=\nicefrac 1{\new{3}0}$, so each section of size $\Delta$ is partitioned into 3 sections of size $\delta$. The $(sell,buy)$ pairs in $S_4$ are as described in Equality \ref{eq:SI} (not all such pairs are shown, there are $\nicefrac{1}\delta =30$ such pairs in $S_4$ ). Note that the gain from trade is $\Delta$ if the bids are $(4\Delta,5\Delta)$ and if a price in between is posted. Also note that  that seller and buyer valuations are equal for $(4\Delta+\delta, 4\Delta+\delta)$ and for $(4\Delta+2\delta, 4\Delta+2\delta)$.}
    \label{fig:my_label}
\end{figure}    
    
    Let $\delta$ and $\Delta$ two positive constants, with $1 \ge \Delta > \delta$ to set later such that $\nicefrac 1\Delta$, $\nicefrac 1\delta$ and $\nicefrac \Delta\delta$ are integers. The grid used in the grid hiding assumption is composed by all integral multiples of $\delta$. For each $i$ from $0$ to $\nicefrac 1\Delta-1$, consider the following sets of valuations:
    \begin{align}
    \nonumber
    S_i &= \Big\{\big(i \cdot \Delta, (i+1) \Delta\big)\Big\} \bigcup_{j \neq i}\bigcup_{k = 0}^{\nicefrac \Delta\delta-1} \Big\{\big(j \cdot \Delta + k \cdot \delta, j \cdot \Delta + (k+1)  \delta \big)\Big\}\\ \label{eq:SI}
      &\quad \bigcup_{k = 1}^{\nicefrac \Delta\delta-1} \Big\{\big(i \cdot \Delta + k \cdot \delta, i \cdot \Delta + \new{(k+1)}\cdot  \delta \big)\Big\}
  \end{align}
    The adversary constructs the first family of sequences as follows: to start, it selects uniformly at random $i$ from $0$ to $\nicefrac 1\Delta-1$, then generates the sequence by repeatedly drawing independently and uniformly at random $(s_t,b_t)$ from $S_i$.
    Note that the cardinality of $S_i$ is $N_\delta:=\nicefrac 1\delta$ for all $i$.
    As a first step, we give a lower bound on the expected gain from trade of the best fixed price in hindsight:
     \new{fix any $S_i$} and any price $p^{\star}_i$ in $(i \cdot \Delta + \delta, (i+1) \cdot \Delta)$. \new{Posting $p^{\star}_i$ when $(s_t,b_t)$ is drawn uniformly at random from $S_i$ yields to three outcomes: (i) with probability $\nicefrac{1}{N_\delta}$ the valuations $(i \cdot \Delta, (i+1) \cdot \Delta)$ are drawn, for a gain from trade of $\Delta$, (ii) with probability $\nicefrac{1}{N_\delta}$, valuations of the type $\big(i \cdot \Delta + k \cdot \delta, i \cdot \Delta + (k+1)\cdot  \delta \big)$ are drawn, with $p_i^* \in \big(i \cdot \Delta + k \cdot \delta, i \cdot \Delta + (k+1)\cdot  \delta \big)$, for a gain from trade of $\delta$, (iii) another pair of valuations is drawn, resulting in no trade. All in all, we have the following:} 
    \begin{align}
    \label{eq:lower_2prices_pstar}
        \E{\max_{p \in [0,1]} \sum_{t=1}^T \gft_t(p)} &\ge \E{\sum_{t=1}^T \gft_t(p^{\star}_i)}= \new{\frac{\Delta + \delta}{N_\delta} T \ge T \delta \Delta}.
    \end{align}
    \new{Note, the expectation is with respect to the sequence of $(s_t,b_t)$ drawn uniformly at random from the fixed $S_i.$}
    Since \new{Inequality~\ref{eq:lower_2prices_pstar}} holds for any realization of the initial choice of $S_i$, it also holds in expectation over all the randomness of the adversary, i.e. also over the random choice of $S_i$.
    
    Let us focus now on the expected performance of any deterministic learner $\A.$ The crux of this proof is that any price that is not a multiple of $\delta$ is not able to discriminate between $S_i$ and $S_j$, for any $j \neq i.$ To see this, let $p$ be any price that is not a multiple of $\delta$, there exist unique $j(p) \in \{0,1, \dots \nicefrac 1\Delta-1\}$ and $k(p)\in \{0,1, \dots \nicefrac \Delta\delta-1\}$ such that
    \[
        p \in \Big(j(p) \cdot \Delta + k(p) \cdot \delta\;, \; j(p) \cdot \Delta + (k(p)+1)\cdot  \delta \Big).    
    \]
    The crucial observation is now that {\em regardless of the $S_i$ selected by the adversary}, the random variable $(\new{\ind{s_t \le p},\ind{p \le b_t}})$ follows the same distribution, in particular, we get:
    \[
        (\new{\ind{S_t \le p},\ind{p \le B_t}}) = 
        \begin{cases}
            (1,1) &\text{ with probability } \frac 1{N_\delta}\\ 
            (1,0) &\text{ with probability } \frac{1}{N_\delta} \left(j(p) \frac \Delta \delta + k(p) - 1\right)\\
            (0,1) &\text{ with the remaining probability}
        \end{cases}    
    \]
    Stated differently, the learner observes a trade with a fixed probability $\nicefrac 1{N_\delta}$, while the probability masses on the left and on the right are determined by the position of $p$, and are constant across all choices of $S_i$ by the adversary. Any trade that the learner observes corresponds to a $\Delta$ gain with probability $\Delta$ and to a $\delta$ gain with the remaining probability. All in all, the gain from trade for any price posted by the learner (in expectation over both the random choice of $S_i$ and the randomness at that specific round) is:
    \[
        \E{\gft_t(p_t)} = \frac{1}{N_\delta} \left[\frac{\Delta}{N_\Delta} + \delta \left(1 - \frac1{N_\Delta}\right)\right]= \delta (\Delta^2 + \delta - \Delta \delta) \le \delta (\Delta^2 + \delta),
    \]
    where $N_{\Delta} := \nicefrac 1\Delta$ and represents the number of possible $S_i$ that the adversary randomly select at the beginning. 
    Summing up the inequality over all times $t$, we get (for all learners that do not post multiples of $\delta$)
    \begin{equation}
        \label{eq:lower_2prices_pstar_2}
        \E{\sum_{t=1}^T \gft_t(p_t)} \le \delta (\Delta^2 + \delta) T   
    \end{equation}
    Via Yao's Theorem and combining \new{Inequalities~\ref{eq:lower_2prices_pstar} and \Cref{eq:lower_2prices_pstar_2}}, we get:
    \[
        R_T^{\alpha,\star} \ge \E{\max_{p \in [0,1]} \sum_{t=1}^T \gft_t(p) - \alpha \sum_{t=1}^T \gft_t(p_t)}
        \ge (\Delta  - \alpha (\Delta^2 + \delta))\delta \cdot T
    \]
    At this point we set\footnote{At the beginning of the proof we assumed $\nicefrac 1\Delta$, $\nicefrac 1\delta$ and $\nicefrac \Delta\delta$ to be integer. It is easy to see that this is without loss of generality, given these choices of $\delta$ and $\Delta$} $\Delta = \nicefrac 1{2\alpha}$ and $\delta = \nicefrac 1{8\alpha^2}$:
    \[
        R_T^{\alpha,\star} \ge \E{\max_{p \in [0,1]} \sum_{t=1}^T \gft_t(p) - \alpha \sum_{t=1}^T \gft_t(p_t)}
        \ge \frac{1}{64 \alpha^3} T.
    \]
    The proof above requires the ``grid hiding" assumption that the learning algorithm cannot post prices that are on the grid  (multiples of $\delta$). 
    
    One way to proceed is to scale down the instance by a constant factor, say $\nicefrac 12$, so that all prices and valuations will be in the range $0$ to $\nicefrac 12$. Then the adversary adds a random uniform number (called a shift) between $0$ up to $\nicefrac 12$ (see \Cref{fig:my_label}). It is clear that any algorithm has zero probability to pinpoint the exact value of the shift (see also the proof of \Cref{thm:lower-two-bits-two-prices}), thus the learning algorithm can post a multiple of $\nicefrac \delta 2$ plus the required shift with probability $0$. In this scaled down instance the total gain from trade derived from the optimal fixed price goes down by a factor of $\nicefrac 12$, while the the gain of the learning algorithm is scaled down by further factor of at least $2$ (since the learner has to deal with the extra uncertainty due to the random shift). Ergo, the $\alpha$-regret will be at least ${\nicefrac{1} {128 \alpha^3}} T$ which completes the proof.
\end{proof}

\subsection{Upper bound on the \texorpdfstring{$2$-regret}\ , posting two prices and given one-bit feedback}\label{subsec:estimator}
    
    The main result in this section is presented in Theorem \ref{thm:upper-one-bit-two-prices}: it is possible to achieve sublinear $2$-regret with one-bit feedback (and by posting two prices). We find this to be the most surprising result in this paper. The crucial ingredient of our approach is an unbiased estimator, $\egft$, of the gain from trade that uses two prices and {\em one single bit} of feedback. This seems quite remarkable. The gain from trade is a discontinuous function composed by two different objects: the difference $(b-s)$ and the indicator variable $\ind{ s \le p \le 
    b}$. Both these two objects are easy to estimate {\em independently}, but for the gain from trade we need an estimator of their product. To estimate $\gft(p)$ for any fixed price $p$, we construct an estimation procedure that considers both features at the same time: it tosses a biased coin with head probability $p$; if head, it posts price $p$ to the buyer and a price drawn uniformly at random in $[0,p]$ to the seller; if tails, it posts price $p$ to the seller and a price drawn u.a.r. in $[p,1]$ to the buyer. The formal procedure is described in the pseudocode, while the following lemma proves that this procedure yields an unbiased estimator of the gain from trade.

    \begin{algorithm*}[t]
        \begin{algorithmic}[t]
        \State \textbf{Input:} price $p$
        \State Toss a biased coin with head probability $p$
        \State \textbf{if} head \textbf{then} Draw $U$ u.a.r. in $[0,p]$ and set $\hat p \gets U$, $\hat q \gets p$
        \State \textbf{else} Draw $V$ u.a.r. in $[p,1]$ and set $\hat p \gets p$, $\hat q \gets V$
        \State Post price $\hat p$ to the seller and $\hat q$ to the buyer and observe the one-bit feedback $\ind{s \le \hat p\le \hat q \le b}$
        \State \textbf{Return:} $\egft(p) \gets \ind{s \le \hat p\le \hat q \le b}$ \Comment{Unbiased estimator of $\gft(p)$}
         \caption*{\textbf{Estimation procedure of $\gft$ using two prices and one-bit feedback}}
    \end{algorithmic}
    \end{algorithm*}

    \begin{lemma}
    \label{lem:estimators}
        Fix any agents' valuations $s,b \in [0,1]$. For any price $p \in [0,1]$, it holds that $\egft(p)$ is an unbiased estimator of $\gft(p)$: $\E{\egft(p)} = \gft(p)$, where the expectation is with respect to the randomness of the estimation procedure.
    \end{lemma}
    
    \begin{proof}
        Note that $p$ is fixed and known to the learner, $s$ and $b$ are fixed but unknown and the learner has to estimate the fixed but unknown quantity $\gft(p) =\ind{s \le p \le q \le b}\cdot (b_t - s_t)$ using only the \new{one}-bit feedback. To analyze the expected value of $\egft(p)$ we define two random variables:
        \[
            X_s(p) = \ind{s \le U \le p \le b}, \ X_b(p) = \ind{s \le p \le V \le b}, 
        \]
        where $U \sim Unif(0,p)$ and $V \sim Unif(p,1)$.
    Clearly, if $p \not \in [s,b]$, the two random variables attains value $0$ with probability $1$ (and therefore are both unbiased estimators of $\gft(p)$ in that case). Consider now the case in which $p \in [s,b]$ and compute their expected value:
    \begin{align*}
        \E{X_s(p)} &= \P{s \le U \le p \le b} = \P{s \le U} = \frac{p-s}{p}, \\
        \E{X_b(p)} &= \P{s \le p \le V \le b} = \P{V \le b} = \frac{b-p}{1-p}.
    \end{align*}
    The estimator $\egft(p)$ works as follows: with probability $p$ it posts prices $(U,p)$, otherwise $(p,V)$, then receives the one-bit feedback from the agents and returns it. Conditioning on the result of the toss of the biased coin it is then easy to compute the expected value of $\egft(p)$:
    \begin{align*}
        \E{\egft(p)} &= p \, \E{X_s(p)} + (1-p) \, \E{X_b(p)} \\
        &=  \ind{s \le p \le q \le b}\, (b - s) = \gft(p).    
    \end{align*}
    \end{proof}

    This estimation procedure becomes a powerful tool to estimate the gain from trade that the learner would have extracted at time $t$ posting price $p$ using randomization and {\em one single bit} of feedback. Note here that the possibility of posting two different prices is crucial: as we have argued in the previous section, one single price is not able to do that, even for two-bit feedback. Given the estimator $\egft$ (actually it consists of a family of estimators: one for each price $p$) we present our learning algorithm \blockdec. Similarly to what is done in Chapter 4 of \citet{NRTV2007}, the learner divides the time horizon in $S$ time blocks $B_{\tau}$ of equal length\footnote{For ease of exposition we assume that $S$ divides $T$. This is without loss of generality in our case, as one can always add some dummy time steps for an additive regret of at most $\nicefrac TS$.} and uses as subroutine some expert algorithm $\calE$ on a meta-instance that considers each time block as a time step and each price in a suitable grid $Q$ as an action. In each block the learner posts the same price $p_{\tau}$ in all but $|Q|$ time steps, where it uses $\egft$ to estimate the total gain from trade obtained in $B_{\tau}$ by all prices in $Q$. The details of \blockdec are presented in the pseudocode.
    
    \begin{algorithm*}[t]
    \begin{algorithmic}[1]
        \State \textbf{Input:} time horizon $T$, number of blocks $S$, grid $Q$ and expert algorithm $\calE$
        \State $\Delta \gets T/S$, $K \gets |Q|$
        \State $B_{\tau} \gets \{({\tau}-1) \cdot \Delta + 1, \dots, \tau\cdot \Delta\}, $ for all $ \tau = 1,2, \dots, S$
        \State Initialize $\calE$ with time horizon $S$ and $K$ actions, one for each $p_i \in Q$
        \For{each round $\tau=1,2,\dots,S$}
        \State Receive from $\calE$ the price $p_{\tau}$
        \State Select uniformly at random an injection $h_{\tau}: Q \to B_{\tau}$ \Comment{We need $\Delta >> |Q|$}
        \For{each round $t\in B_{\tau}$}
        \If{$h_{\tau}(p_i) = t$ for some price $p_i$}
            \State Use the estimator $\egft(p_i)$ at time $t$ and call its output $\egft_{\tau}(p_i)$
        \EndIf
        \State \textbf{else:} Post price $p_{\tau}$ and ignore feedback
        \EndFor
        \State Feed to $\calE$ the estimated gains $\{ \egft_{\tau}(p_i)\}_{i=1,\dots, K}$
        \EndFor
    \end{algorithmic}
    \caption*{\blockdec (\bd)}\label{alg:block}
    \end{algorithm*}
    
    Consider any instantiation of the algorithm \blockdec, fix any block $B_{\tau}$ and price $p$. With a slight abuse of notation we denote the average gain from trade posting price $p$ in $B_{\tau}$ as $\gft_{\tau}$; formally,
    \[
        \gft_{\tau}(p) = \frac{1}{\Delta} \sum_{t \in B_{\tau}} \gft_t(p).
    \]
    We show that $\egft_{\tau}(p)$ as defined in the pseudocode is an unbiased estimator of $\gft_{\tau}(p)$, where the randomization is due to the random choice of the injective function $h_{\tau}$ and the inherent randomness in the estimator $\egft$. 
    \begin{lemma}     \label{lemma:second_estimator}
        Fix any sequence of valuations, then the random variable $\egft_{\tau}(p_i)$ is an unbiased estimator of $\gft_{\tau}(p_i)$ for any $\tau \in \{1,2,\dots,S\}$ and price $p_i$ on the grid $Q$ .
    \end{lemma}
    \begin{proof}
        For any fixed price $p_i$ it is clear that $h_{\tau}(p_i)$ is distributed uniformly at random in the time steps contained in the block $B_{\tau}$. Moreover, given $h_{\tau}$, the $\egft$ are still unbiased estimators of the corresponding time steps. Thus, we have the following:
        \begin{align*}
            \E{\egft_\tau(p_i)} &= \sum_{t \in B_{\tau}} \P{h_{\tau}(p_i) = t} \E{\egft_t(p_i) \mid h_{\tau}(p_i) = t}   \\
            &= \sum_{t=1}^T \frac1{\Delta} \E{\egft_
            {t}(p_i)\mid h_{\tau}(p_i) = t}\\
            &= \sum_{t=1}^T \frac1{\Delta} \gft_t(p_i) = \gft_{\tau}(p_i).
        \end{align*}
        A notational observation: with the random variable $\egft_t(p)$ we refer to the result of the estimation procedure in $p$ run at time $t$, which is an unbiased estimator of the gain from trade of price $p$ achievable at time $t$.
    \end{proof}

\begin{theorem}[Upper bound on $2$-regret posting two prices, one-bit feedback]   
\label{thm:upper-one-bit-two-prices}
    In one-bit feedback model where the learner is allowed to post two prices, the $2$-regret of \blockdec (\bd) is such that $
        R_T^2(\bd) \le 5 T^{\nicefrac 34} \sqrt{\log(T)},$
    for appropriate choices of the expert algorithm $\calE$, grid $Q$ and number of blocks $S.$
\end{theorem}

\begin{proof}
    We consider a grid $Q$ of equally spaced prices (we set the step later), and denote with $\Delta = \nicefrac TS$ the length of every time block. The learner keeps playing the same price in each block, apart from the explorations steps, that are drawn uniformly at random. The learner decides which action to play according to some routine $\calE$ that is run on $S$ time steps and $|Q|$ actions: this is the reason we talk interchangeably of actions and prices. 
    
   From Lemma \ref{lemma:second_estimator} we know that the estimators in a block, i.e., $\egft_{\tau}(p_i)$ are indeed unbiased estimators of $\gft_{\tau}(p_i) $. Since this holds for any price $p_i \in Q$, it also holds for any random price $\hat p$ whose randomness is independent on the choice of the injection $h_{\tau}$ and the internal randomization of the estimators. Thus, the same holds even if instead of a fixed price $p_i$ we consider price  $p_{\tau}$ posted by the algorithm because it depends {\em only} on what happened in past blocks.
    Let now $\calE$ be the Multiplicative Weights algorithm \new{\citep{AroraHK12}}. If we fix the randomness in the exploration and in the estimation upfront and consider only the inherent randomness in $\calE$ we inherit the bound on the regret of $\calE$ on the realized estimated gain from trades (note that they are all bounded in $[0,1]$)
    \begin{equation}
    \label{eq:regret_calE}
        \max_{p\in Q}\sum_{\tau=1}^{S} \egft_{\tau}(p) - \E{\sum_{\tau = 1}^S \egft_{\tau}(p_{\tau})} \le 2 \sqrt{S\log(|Q|)}
    \end{equation}
    The randomness of $\calE$ depends somehow on the realizations of the random injections and estimators, but if we look at any block $B_\tau$, we see that the random price output by the routine is independent from $h_{\tau}$ and the estimators in that block. Therefore, we can safely take the expected value (on the randomness of the $h_{\tau}$ and the estimators) on both sides of \new{Inequality~\ref{eq:regret_calE}, and start the following chain of inequalities:
    \begin{align}
    \nonumber
        2 \sqrt{S\log(|Q|)} &\ge \E{\max_{p\in Q}\sum_{\tau=1}^{S} \egft_{\tau}(p) - \sum_{\tau = 1}^S \egft_{\tau}(p_{\tau})}\\&\ge \max_{p\in Q}\E{\sum_{\tau=1}^{S} \egft_{\tau}(p) - \sum_{\tau = 1}^S \egft_{\tau}(p_{\tau})} \tag{Jensen Inequality}\\
        \label{eq:regret_calE_exp}
        &= \max_{p\in Q}\sum_{\tau=1}^{S} \gft_{\tau}(p) - \E{\sum_{\tau = 1}^S \gft_{\tau}(p_{\tau})},
    \end{align}
     where the equality follows from \Cref{lemma:second_estimator}.}
    Now, we move from the blocks time scale to the normal one and multiply everything by a factor $\Delta$. Our algorithm does not always play $p_{\tau}$, but for each one of the block, it spends $|Q|$ steps exploring. Therefore, we need to consider an extra $|Q|S$ additive term. 
    At this point, we have all the ingredients to bound the $2$-regret of our algorithm. Plugging \new{Inequality~\ref{eq:regret_calE_exp}} and the observation about the extra losses incurred by the exploration into the discretization inequality (\Cref{claim:discretization}) we get:
    \[
        R_T^2(\bd) \le 2\Delta \sqrt{{S}\log{|Q|}} + |Q|S + \delta(Q)T.
    \]
    The theorem then follows by optimizing the free parameters: we set $\Delta = \sqrt{T}$ and we choose $Q$ to be the uniform grid of multiples of $T^{-\nicefrac 14}$ (thus $S= \sqrt{T}$, $\delta(Q) = T^{-\nicefrac 14}$ and $|Q| = T^{\nicefrac 14}+1$). 
    \end{proof}

\subsection{Lower bound on \texorpdfstring{$2$-regret}\ , posting two prices and two-bit feedback}

    In this section, we complement the positive results for the single price and two-bit feedback setting with a lower bound on the $2$-regret achievable in the (easier) two prices setting. The lower bound construction entails a careful combination of ``price hiding'' and an exploration-exploitation trade off that is inspired to the partial monitoring literature (more precisely to the revealing action game, see e.g., \citet{BartokFPRS14}). First, in \Cref{lem:two_thirds} we construct a randomized sequence of valuations such that any deterministic learning algorithm that is restricted to {\em not} post three specific pair of prices suffers at least $\Omega(T^{\nicefrac 23})$ expected $2$-regret (against the best fixed price in hindsight that possibly uses these forbidden prices). Then, in \Cref{thm:lower-two-bits-two-prices} we derive the general lower bound by hiding the forbidden prices with a random shift (as in \Cref{thm:lower-two-bits-one-price}). 

    We start by constructing a randomized instance that induces $\Omega(T^{\nicefrac 23})$ $2$-regret against all the deterministic algorithms that are forbidden to post prices $(p_i,p_i)$ for $p_0 = \nicefrac 12$, $p_1 = \nicefrac 14,$ and $p_2= \nicefrac 34.$ This instance is a uniform mixture of three stationary distributions, a baseline and two perturbations.
    
    \paragraph{The family of randomized sequences} The family of distributions we introduce are supported on the same set of points \supp and are parametrized by $\eps \in (-\nicefrac 1{12}, \nicefrac 1{12})$. The support \supp is composed by four high-gain-from-trade valuations (with $b-s = \nicefrac 14$), i.e.,  $(0, \nicefrac 14),(\nicefrac 14, \nicefrac 12),(\nicefrac 12, \nicefrac 34)$, and $(\nicefrac 34,1)$, and four low-gain-from-trade ones (with $b-s = \nicefrac 1{12}$), i.e., $(\nicefrac 16, \nicefrac 14), (\nicefrac 14,\nicefrac 13), (\nicefrac 23, \nicefrac 34),$ and $(\nicefrac 34, \nicefrac 56)$. The valuations $(0, \nicefrac 14)$ and $(\nicefrac 34, \nicefrac 56) $ are drawn with probability $\nicefrac 18 + \eps$, the valuations $(\nicefrac 16, \nicefrac 14)$ and $(\nicefrac 34,1)$ with probability $\nicefrac 18 - \eps$, while the remaining valuations in the support are drawn with probability  $\nicefrac 18$. As a convention, we denote with $\Pb^{\eps}$ the probability measure under which the random variables $(S_t,B_t)$ are drawn i.i.d. according to the distribution corresponding to parameter $\eps$ we just described (we use $\mathbb E^{\eps}$ to denote the relative expectation). To avoid confusion with the explicit dependence on the time step $t$, we refer to $(S,B)$ as a generic random variable drawn according to the distribution, and with $\gft$ the corresponding gain from trade. The expected gain from trade corresponding to posting any pair of budget balanced $(p,q)$ prices under $\Pb^{\eps}$ can be computed explicitly:
    \[ 
    \mathbb E^{\eps}[\gft(p,q)] = 
        \begin{cases}
            \tfrac 1{32} + \tfrac \eps 4  &\text{ if } (p,q) \in E_1\\
            \tfrac 1{32} - \tfrac \eps 4 &\text{ if } (p,q) \in E_2\\
            \tfrac 1{32} &\text{ if } (p,q) \in E\\
            \tfrac 1{24} + \tfrac \eps 6 &\text{ if } (p,q) \in G_1\\
            \tfrac 1{24} - \tfrac \eps 6 &\text{ if } (p,q) \in G_2\\
            \tfrac 1{24} &\text{ if } (p,q) \in G\\
            \tfrac 1{16} &\text{ if } (p,q) = (\tfrac 12, \tfrac 12) \\
            \tfrac 1{12} + \tfrac \eps 6 &\text{ if } (p,q) = (\tfrac 14, \tfrac 14)\\
            \tfrac 1{12} - \tfrac \eps 6 &\text{ if } (p,q) = (\tfrac 34, \tfrac 34) \\
            0 &\text{ for all the others $(p,q) \in B$} 
        \end{cases}
    \]
    We denote with $\U$ the set of all the budget balanced prices ($p \le q$) which corresponds to the upper-left triangle in the $[0,1]^2$ square. The other two families of sets are defined as follows (see also  \Cref{fig:support}). The $E_i$ sets are $E_1 = [0,\tfrac 16) \times [0,\tfrac 14] \cap \, \U$, $E_2 = [\tfrac 34,1] \times (\tfrac 56,1] \cap  \, \U$, and $E = \left([\tfrac 14,\tfrac 12) \times (\tfrac 13,\tfrac 12] \cup [\tfrac 12,\tfrac 23) \times (\tfrac 12, \tfrac 34]\right) \cap  \, \U$. The $G_i$ sets are $G_1 = [\tfrac 16,\tfrac 14) \times [\tfrac 16,\tfrac 14] \cap  \, \U$, $G_2 = [\tfrac 34,\tfrac 56] \times (\tfrac 34,\tfrac 56] \cap  \, \U$, and $G = \left([\tfrac 14,\tfrac 13] \times (\tfrac 14,\tfrac 13] \cup [\tfrac 23,\tfrac 34) \times [\tfrac 23,\tfrac 34]\right) \cap  \, \U$.
    
    \begin{figure}
        \centering
        \begin{subfigure}{.45\textwidth}
        \centering
            \includegraphics[width = 1\textwidth]{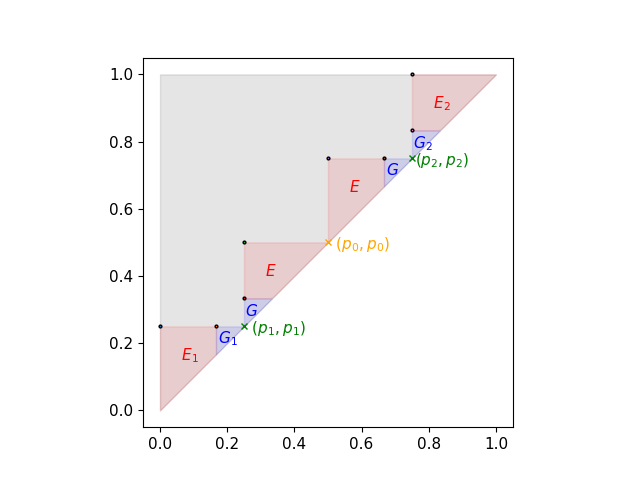}
        \caption{$\supp$ and gain from trade regions.}
        \label{fig:support}    
        \end{subfigure}
        \begin{subfigure}{.45\textwidth}
        \centering
        \includegraphics[width = 1\textwidth]{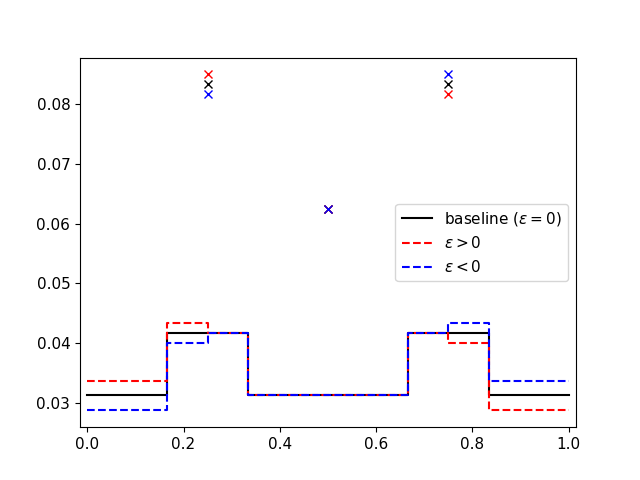}
        \caption{Expected Gain From Trade: $
        \mathbb E^{\eps} [\gft(p)]$}
        \label{fig:gft_profile}    
        \end{subfigure}
        \caption{Visualization of the family of distributions used in the lower bound construction of \Cref{thm:lower-two-bits-two-prices}.}
    \end{figure}

    The expected gain from trade function exhibits a rich behaviour (see \Cref{fig:gft_profile} for a visualization of the expected gain from trade profile achievable by posting the same price to both agents). The set of budget balanced prices $\U$ is naturally partitioned into $5$ regions which are visualized in \Cref{fig:support}. Prices in the same region share the same expected gain from trade (up to an additive $\Theta(\eps)$ term), while prices in different regions yield an expected gain from trade that is an additive constant factor apart. Precisely, we have the following:
    \begin{itemize}
        \item There are two ``optimal'' (strongly budget balanced) prices, $p_1 = \nicefrac 14$ and $p_2 = \nicefrac 34$ that yield an expected gain from trade around $\nicefrac 1{12}$ (green in \Cref{fig:support}). If $\varepsilon >0$, then $p_1$ is better than $p_2$ (by an additive $\Theta(\eps)$ factor), if $\eps <0 $ the converse holds.
        \item There is an isolated (strongly budget balanced) price  $p_0 = \nicefrac 12$ that guarantees an expected gain from trade of $\nicefrac 1{16},$ regardless of $\eps$ (orange in \Cref{fig:support}).
        \item The are three ``good'' regions $G_1, G_2$ and $G$ that yield an expected gain from trade around $\nicefrac 1{24}$ (blue in \Cref{fig:support}). While $G$ guarantees exactly $\nicefrac 1{24}$ regardless of $\eps$, the two $G_i$ are such that one is an additive $\Theta(\eps)$ term larger and another is an additive $\Theta(\eps)$ term smaller (depending on the sign of $\eps$).
        \item There are three ``bad'' regions $E_1, E_2$ and $E$ that yield an expected gain from trade around $\nicefrac 1{32}$ (red in \Cref{fig:support}). While $E$ guarantees exactly $\nicefrac 1{32}$ regardless of $\eps$, the two $E_i$ are such that one is an additive $\Theta(\eps)$ term larger and another is an additive $\Theta(\eps)$ term smaller (depending on the sign of $\eps$).
        \item All other prices yield $0$ gain from trade (gray in \Cref{fig:support}).
    \end{itemize} 

        \paragraph{The construction of the hard randomized sequence} The hard (randomized) sequence $\mathcal S^*$ is constructed as follows: with probability $\nicefrac 13$ the sequence of valuations is drawn i.i.d. according to probability measure $\Pb^0$, with probability $\nicefrac 13$ it is drawn i.i.d. according to $\Pb^{\eps}$ for some positive ${\eps}$ we set later, and with the remaining probability according to $\Pb^{-\eps}$ (for the same positive $\eps$). To avoid overloading the notation, we refer to the latter two probabilities measures with $\Pb^1$, respectively $\Pb^2$. Note, the optimal (ex-ante) fixed price under $\Pb^1$ is $p_1$ (posted to both agents), under $\Pb^2$ is $p_2 $ (posted to both agents); while under $\Pb^0$ both these prices are optimal. As a first quantitative result we prove that any deterministic learning algorithm that is constrained so as not to post any of the prices $p_0,p_1,$ and $p_2$ suffers $\Omega(T^{\nicefrac 23})$ $2$-regret against $\S^*$ (such a ``constraint'' is effectively achieved by hiding these prices, see \Cref{thm:lower-two-bits-two-prices}).
        
        \begin{lemma}
        \label{lem:two_thirds}
            Any deterministic learning algorithm $\mathcal A$ that cannot post prices $(p_i,p_i)$ for $i=0,1,2$ suffers expected $2$-regret against $\mathcal S^*$ (for a certain $0 < \eps <\nicefrac 1{12}$) that is at least:
            \[
                \E{R_T^2(\mathcal A, \mathcal S^*)} \ge \frac{1}{144} T^{\nicefrac 23}.
            \]
        \end{lemma} 
        \begin{proof}
            We start by studying the feedback that different prices receive. Crucially, the only place where the learner can hope to distinguish between the three distributions is in the exploration region $X = ([0, \tfrac 16) \times [0,1] \cup [\tfrac 16, 1] \times (\tfrac 56,1]) \cap \U$ (red in \Cref{fig:feedback_regions}). In fact, the only possibility to observe the $\eps$ difference in the probability measures is to observe either one of the two long intervals $(0,\nicefrac 14)$ or $(\nicefrac 34, 1)$ (with probability $\nicefrac 18 \pm \eps$) while avoiding the corresponding short interval (respectively $(\nicefrac 16,\nicefrac 14)$ or $(\nicefrac 34, \nicefrac 56),$ with probability $(\nicefrac 18 \mp \eps$). Posting prices in all the other regions does not help, as formalized in the following \new{c}laim. 
            \begin{figure}
                \centering
                \includegraphics[width = 0.5\textwidth]{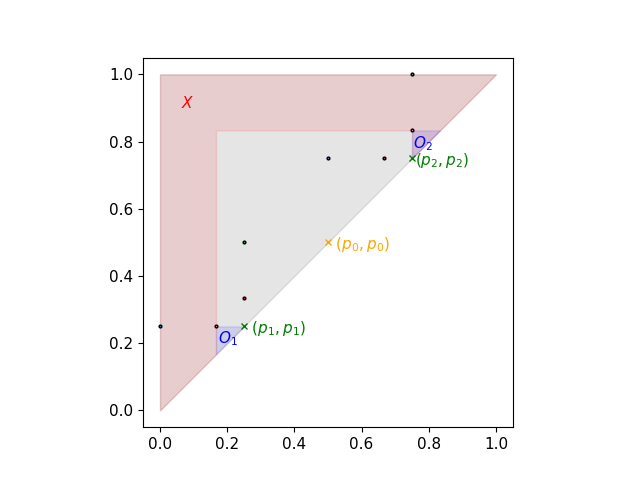}
                \caption{Feedback regions}
                \label{fig:feedback_regions}
            \end{figure}
            \begin{claim}
                \label{cl:no_exploration}
                Let $Z$ be the random variable representing the feedback received by posting prices $(p,q) \in [0,1]^2 \setminus X$, then for any $z \in \{0,1\}^2$ the following holds:
                \[
                    \Pb^0(Z = z) = \Pb^1(Z = z) = \Pb^2(Z = z).
                \]
            \end{claim}
            \begin{proof}[Proof of \Cref{cl:no_exploration}]
                For any $(p,q)$, the feedback $Z$ received is identified by the probability that $(S,B)$ falls within the following four rectangles that partition the $[0,1]^2$ square: upper-left $Q^{(1,1)}_{p,q} = [0,p] \times [q,1]$, upper-right $Q^{(0,1)}_{p,q} = (p,1] \times [q,1]$, lower-left $Q^{(1,0)}_{p,q} = [0,p] \times [0,q)$, and lower-right $Q^{(0,0)}_{p,q} = (p,1] \times [0,q)$. Formally, for any $z \in \{0,1\}^2$ we have that
                \[
                    \Pi{Z = z} = \Pi{(S,B) \in Q^z_{p,q}} = \sum_{(s,b) \in \supp \cap Q^z_{p,q}} \P{(S,B) = (s,b)},
                \]
                where \supp is the eight-point support of the distributions. Now, the only points whose probabilities change for different $\Pb^i$ are the ones contained in $X$, and are such that their variations cancel out. More precisely, for any $i=0,1,2$ it holds that:
                \[
                \Pi{(S,B) \in \{(0,\tfrac{1}{4}),(\tfrac 16,\tfrac{1}{4})\}} = \Pi{(S,B) \in \{(\tfrac{3}{4},1),(\tfrac 56,1)\}} = \tfrac 14.
                \]
                Therefore, for any $(p,q) \notin X$ it holds that
                \[
                    \Pi{Z = z} = \frac 18 \sum_{(s,b) \in \supp} \ind{(s,b)\in Q^z_{p,q}}.    
                \]
                Since the latter term does not depend on $i$, this concludes the proof of the \new{c}laim.
            \end{proof}
            \Cref{cl:no_exploration} and the analytic characterization of the expected gain from trade for different prices induce the following natural partition of $\U$ (see also \Cref{fig:feedback_regions}):
            \begin{itemize}
                \item \textit{The exploration region $X$}. Playing any pair of prices $(p,q) \in X$ the following inequality holds for $i = 1,2$:
                \begin{equation}
                    \label{eq:cost_expl}
                    \Ei{\gft(p_i) - 2 \gft(p,q)} \ge (\tfrac 1{12} + \tfrac \eps 6) - 2 (\tfrac 1{32} + \tfrac \eps 4)\ge \tfrac 1{48 }- \tfrac \eps {12}.
                \end{equation}
                Under $\mathbb E^0$ an analogous inequality holds, but without the $\eps$ term: 
                \begin{equation}
                    \label{eq:cost_expl_0}
                    \Ez{\gft(p_1) - 2 \gft(p,q)} = \Ez{\gft(p_2) - 2 \gft(p,q)} \ge \tfrac 1{48 }.
                \end{equation}
                \item \textit{The two candidate  regions $O_1 = G_1$ and $O_2 = G_2$}. When $(p,q)$ is in $O_i$ under $\mathbb E^i$ the following inequalities hold for $i=1,2$:
                \begin{equation}
                \label{eq:gain_opt}
                    \Ei{\gft(p_j) - 2 \gft(p,q)} \ge (\tfrac 1{12} + \tfrac \eps 6) - 2 (\tfrac 1{24} + \tfrac \eps 6)\ge  - \tfrac \eps 6.    
                \end{equation}
                Conversely, posting prices in  $O_j$ ($j \neq i$) under $\mathbb E^i$ yields:
                \begin{equation}
                \label{eq:cost_subopt}
                    \Ei{\gft(p_i) - 2 \gft(p,q)}\ge (\tfrac 1{12} + \tfrac \eps 6) - 2 (\tfrac 1{24} - \tfrac \eps 6)\ge \tfrac \eps 2. 
                \end{equation}
                Under $\mathbb E^0$ there is no difference between $O_1$ and $O_2$: for any $(p,q)$ in $O_1 \cup O_2$ it holds that
                \begin{equation}
                \label{eq:cost_opt_0}
                    \Ez{\gft(p_1) - 2 \gft(p,q)}=\Ez{\gft(p_2) - 2 \gft(p,q)}=0. 
                \end{equation}
                
                \item The rest of $\U\setminus \{(p_0,p_0), (p_1,p_1), (p_2,p_2)\}$. There the largest achievable gain from trade is $\nicefrac 1{24}$ (by posting prices in $G$), so the $2$-regret is as follows for $i=1,2$
                \begin{align}
                    \label{eq:cost_neutrality}      \Ei{\gft(p_i) - 2 \gft(p,q)} &\ge (\tfrac 1{12} + \tfrac \eps 6) - 2 \tfrac 1{24} \ge \tfrac \eps 6 \\
                    \label{eq:cost_neutrality_0}
                    \Ez{\gft(p_1) - 2 \gft(p,q)} &= \Ez{\gft(p_2) - 2 \gft(p,q)} =0 
                \end{align}
            \end{itemize}

    This partition into exploration region, candidate regions and rest of $\U$ simplifies the analysis of the expected performance of $\mathcal A$. Let $M_T(1)$, respectively $M_T(2)$, be the random variables that count the number of times that prices are posted in $O_1$, respectively $M_T(2)$, and let $N_T$ be the one that counts the prices posted in $X.$ We start by lower bounding the $2$-regret of $\mathcal A$ against $\Pb^1$:
    \begin{align}
    \nonumber
        &\mathbb E^1\left[\max_{p \in [0,1]} \sum_{t=1}^T \gft_t(p) - 2 \sum_{t=1}^T \gft_t(p_t,q_t)\right]\\
        &\ge \Eo{ \sum_{t=1}^T \gft_t(p_1) - 2 \sum_{t=1}^T \gft_t(p_t,q_t)}\tag*{(Jensen's inequality)}\\
    \nonumber
        &\ge \mathbb E^1\Big[\underbrace{(\tfrac 1
    {48 }- \tfrac \eps {12})N_T}_{\text{\new{Inequality~\ref{eq:cost_expl}}}}  \underbrace{- \tfrac \eps 6 M_T(1)}_{\text{\new{Inequality~\ref{eq:gain_opt}}}} + \underbrace{\tfrac \eps 2 M_T(2)}_{\text{\new{Inequality~\ref{eq:cost_subopt}}}} + \underbrace{\tfrac \eps 6 (T - M_T(2) - M_T(1) - N_T)}_{\text{\new{Inequality~\ref{eq:cost_neutrality}}}}\Big]\\
    \nonumber
    &= 
    \Eo{(\tfrac 1
    {48 }- \tfrac \eps {4})N_T + \tfrac \eps 3 M_T(2) - \tfrac \eps 3 M_T(1) + \tfrac \eps 6 T}\\
    \label{eq:intermediate_i}
    &\ge \Eo{\tfrac \eps 3 M_T(2) - \tfrac \eps 3 M_T(1)} + \tfrac \eps 6 T,
    \end{align}
    where the last inequality holds for any $\eps \in (0,\nicefrac 1{12})$.
    Now, the two perturbed distributions coincide with the baseline but for a small additive perturbation on some points of the support. The crucial observation is that the only way the learner can discriminate between scenarios is by posting prices in the exploration region (\Cref{cl:no_exploration})! This implies that the history of the algorithm does not differ much under $\mathbb{E}^0$, $\mathbb{E}^1$, and $\mathbb{E}^2$, up to an additive term that depends on the number of times the algorithm posted prices in the exploration region (and thus observed different feedback). This intuition is formalized by the following \new{c}laim.
    \begin{claim}
    \label{cl:TV}
        For any $\eps \in (0,\tfrac 1{12})$, the following inequality holds true for all $i,j \in \{1,2\}$, with $i \neq j$:
        \[
            |\Ei{M_T(j)} - \Ez{M_T(j)}| \le 4\eps T \sqrt{\Ez{N_T}}
        \]
    \end{claim}
    \begin{proof}[Proof of \Cref{cl:TV}]
        For any $t = 1,2,\dots,T$, the random prices $(P_t,Q_t)$ selected by $\A$ at round $t$ are a deterministic function of the random feedback $Z_1, \dots, Z_{t-1}$, received in the previous $t-1$ time steps. In formula, for any $i,j$ as in the statement, we then have the following
        \begin{align}
            \nonumber
            |\Ei{ M_T(j) } - \Ez{ M_T(j) }| &=
            \babs{\sum_{t = 2}^T \Pi{ (P_t,Q_t) \in O^j} - \Pz{ (P_t,Q_t) \in O^j}}\\&\le
            \sum_{t = 2}^T \babs{\Pi{ (P_t,Q_t) \in O^j} - \Pz{ (P_t,Q_t) \in O^j}}\\
        \label{eq:TV}&\le
            \sum_{t = 2}^T \bno{ \mathbb P^i_{(Z_1, \dots, Z_{t-1} )} - \mathbb P^0_{(Z_1, \dots, Z_{t-1} )}}_{\mathrm{TV}},
        \end{align}
        where $||{\cdot}||_{\mathrm{TV}}$ denotes the total variation norm and $P^i_{(Z_1, \dots, Z_{t-1} )}$ is the push-forward measure (on $(\{0,1\}^2)^{t-1}$) induced by the execution history of the algorithm up to time $t-1$ when the valuations are drawn according to $\mathbb P^i$, for $i \in \{ 0,1,2\}$. In \Cref{claim:KL} in the Appendix we prove via the Pinsker inequality and a careful case-analysis that the following inequality holds for any $t$ and $i$:
        \[
            \bno{ \mathbb P^i_{(Z_1, \dots, Z_{t-1} )} - \mathbb P^0_{(Z_1, \dots, Z_{t-1} )}}_{\mathrm{TV}} \le 4 \eps \sqrt{\Ez{N_T}}.    
        \]    
        Combining this inequality with \new{Inequality~\ref{eq:TV}} yields the desired bound.     
    \end{proof}

    \Cref{cl:TV} allows us to lower bound the right-hand side of \new{Inequality~\ref{eq:intermediate_i}} with something that does not depend on the perturbed distribution $\Pb^1$, but on the baseline $\Pb^0$: 

    \begin{align}
    \nonumber
        \mathbb E^1&\left[\max_{p \in [0,1]} \sum_{t=1}^T \gft_t(p) - 2 \sum_{t=1}^T \gft_t(p_t,q_t)\right]\\
        &\ge \Eo{\tfrac \eps 3 M_T(2) - \tfrac \eps 3 M_T(1)} + \tfrac \eps 6 T\tag{by \new{Inequality~\ref{eq:intermediate_i}}}\\
        \label{eq:final_o}
        &\ge \Ez{\tfrac \eps 3 M_T(2) - \tfrac \eps 3 M_T(1)} + \tfrac \eps 6 T -  \tfrac{4}{3} \eps^2 T \sqrt{\Ez{N_T}}
    \end{align}
    
    The analogue of \new{Inequality~\ref{eq:intermediate_i}} for probability measure $\Pb^2$ can be similarly derived and, together with $\Cref{cl:TV}$, it gives the analogue of \new{Inequality~\ref{eq:final_o}} for $\Pb^2$: 
    \begin{align}
    \nonumber
        \mathbb E^2&\left[\max_{p \in [0,1]} \sum_{t=1}^T \gft_t(p) - 2 \sum_{t=1}^T \gft_t(p_t,q_t)\right]\\
    \label{eq:final_t}
    &\ge \Ez{\tfrac \eps 3 M_T(1) - \tfrac \eps 3 M_T(2)} + \tfrac \eps 6 T -  \tfrac{4}{3} \eps^2 T \sqrt{\Ez{N_T}}
    \end{align}

    When the learning algorithm $\A$ faces the baseline distribution $\Pb^0$, then it never suffers negative expected $2$-regret, in particular by \new{Inequalities~\ref{eq:cost_expl_0} and \ref{eq:cost_opt_0}, and Equality \ref{eq:cost_neutrality_0},} we have the following:
    \begin{equation}
        \label{eq:final_0}
        \mathbb E^0\left[\max_{p \in [0,1]} \sum_{t=1}^T \gft_t(p) - 2 \sum_{t=1}^T \gft_t(p_t,q_t)\right]\ge \tfrac{1}{48}\Ez{N_T}
    \end{equation}

    Averaging the three contributions (i.e., considering the expected performance of the deterministic algorithm against the uniform mixture of the three distributions) we have the following: 
    \begin{align*}
        \E{R_T^2(\A,\S^*)} &= \frac 13 \sum_{i=0,1,2}\mathbb E^i\left[\max_{p \in [0,1]} \sum_{t=1}^T \gft_t(p) - 2 \sum_{t=1}^T \gft_t(p_t,q_t)\right]\\
        &\ge 
        \frac \eps 9 T -\frac{8}{9} \eps^2 T \sqrt{\Ez{N_T}} + \frac{1}{144}\Ez{N_T}\tag{\new{Ineq. from \ref{eq:final_o} to \ref{eq:final_0}}}\\
        &\ge 
        \frac \eps 9 T + \frac{1}{144}\sqrt{\Ez{N_T}} \left(\sqrt{\Ez{N_T}} - 128 \eps^2 T\right)\tag*{\text{(minimized in $\sqrt{\Ez{N_T}}= 64 \eps^2 T$})}\\
        &= \frac \eps 9 T \left(1 - 256
        \eps^3 T\right) \ge \frac{1}{144} T^{\nicefrac 23}. \tag*{\text{(for $\eps = \tfrac 1{8}T^{-\nicefrac 13}$)}}
    \end{align*}
    Note, we carried calculations by assuming $\eps \in (0,\nicefrac 1{12})$, so in order for it to be compatible with the specific value of $\eps$ we choose we need to assume $T \ge 4$, this (together with the fact that the inequality in the statement is trivially verified for $T<4$) concludes the proof of \Cref{lem:two_thirds}.
    \end{proof}

    To conclude our lower bound, we need to show how to {\em hide} the three prices $p_0,p_1, $ and $p_2$. To this end, it is sufficient to add a small, random, perturbation to $\S^*$. 
    \begin{theorem}[Lower bound on $2$-regret for two prices and two-bit feedback]   
    \label{thm:lower-two-bits-two-prices}
        In the two-bit feedback model where the learner is allowed to post two prices, for all horizons $T$, the minimax $2$-regret satisfies \[
            R_T^{2,\star}\ge \frac{1}{288} T^{\nicefrac 23}.
        \]
    \end{theorem}
    \begin{proof}
        We prove this result via Yao's principle, by showing the existence of a randomized sequence of valuations against which any deterministic algorithm suffers, on average, the desired $2$-regret. Consider an adversary that constructs a random sequence $\overline{\mathcal S}$ via the following procedure: the adversary draws an uniform random variable $U \in [0,\nicefrac 1{2}]$ up-front and then generates a sequence of valuations $(S_t,B_t)$ from $\mathcal{S}^*$. The sequence actually posted by the adversary is as follows:
        \[
            \overline{S}_t = \tfrac 12 \cdot S_t + U, \overline{B}_t = \tfrac 12 \cdot B_t + U. 
        \]
        Consider any fixed deterministic algorithm $\mathcal A$. At each iteration, it observes one of four feedback signals (one in $\{0,1\}^2$), therefore the set of all the possible prices $P_{\mathcal A}$ it posts is discrete and finite. Let now $\mathcal E$ be the event that the realized $U=u$ is such that $(\nicefrac {p_1}2 + u, \nicefrac {p_1}2 + u), (\nicefrac {p_0}2 + u, \nicefrac {p_0}2 + u),$ and $(\nicefrac {p_2}2 + u, \nicefrac {p_2}2  + u)$ do not belong to $P_{\mathcal A}$; clearly $\P{\cal E} = 1$. Finally, let $\textbf A$ be the family of all the deterministic algorithms that can post two prices, receive two-bit feedback but cannot post prices $p_0,p_1,p_2$, we have the following:
        \begin{align*}
            \E{R_T^2(\mathcal A,\overline{\mathcal S})} &= \E{R_T^2(\mathcal A,\overline{\mathcal S})\mid \cal E} \P{\cal E} + \E{R_T^2(\mathcal A,\overline{\mathcal S})\mid \cal E^C}\P{\cal E^C}\\
            &=\E{R_T^2(\mathcal A,\overline{\mathcal S})\mid \cal E} \tag{\text{$\P{\cal E} = 1$}}\\
            &\ge \frac 12 \sup_{\hat \A \in \textbf A} \E{R_T^2(\hat \A,{\mathcal S^*})}\\
            &\ge \frac{1}{288} T^{\nicefrac 23}. \tag{by \Cref{lem:two_thirds}}
        \end{align*}
        Note, the first inequality follows by the observation that the $2$-regret suffered by $\A$ against $\overline{\mathcal S}$ is at least the $2$-regret suffered by the best deterministic algorithm that has the same constraint as in \Cref{lem:two_thirds}, multiplied by $\nicefrac 12$ because the original hard instance is scaled by that factor. 
        \end{proof}

\section{Discussion, Extensions, and Open Problems}

     In this paper, we investigate the sequential bilateral trade problem with adversarial valuations. We study various feedback scenarios and consider the possibility for the mechanism to post one price vs. when it can post different prices for buyer and seller. We identify the exact threshold of $\alpha$ that allows sublinear $\alpha$-regret. We show that with a partial feedback it is impossible to achieve sublinear $\alpha$-regret for any constant $\alpha$ with a single price while $2$-regret is achievable with $2$ prices. Finally, we show a separation in the minimax $2$-regret between full and partial feedback. Although in this paper we only consider the gain from trade, our positive results trivially also hold with respect to social welfare. Furthemore, modifying our lower bound from \Cref{thm:lower-full-2-eps} it is possible to show that sublinear $\alpha$-regret is not achievable for $\alpha < 2$ with respect to social welfare. An obvious open problem, with respect to both gain from trade and to social welfare, consists in determining the exact regret term as a function of $T$. Clearly there is a gap in our Table of results, and the exact term is yet unclear also for social welfare. We focus on the sequential problem where at each step one buyer and one seller appears. It would be interesting to study the model where multiple buyers and multiple sellers arrive at each time step and sellers have values for their goods, buyers have values for the different goods. \new{Finally, it is natural to adapt the bilateral trade model so to capture more complex features arising from the ridesharing platform applications, such as time duration, contextual information (e.g., location, destination, type of car), and the interplay of different drivers and users interacting at the same time.}

\section*{Acknowledgement}
Yossi Azar was supported in part by the Israel Science Foundation (grant No. 2304/20). Federico Fusco was supported by the ERC Advanced
Grant 788893 AMDROMA ``Algorithmic and Mechanism
Design Research in Online Markets'', PNRR MUR project
``PE0000013-FAIR'', and PNRR MUR project ``IR0000013-
SoBigData.it''. The authors would like to thank Roberto Colomboni for pointing out an issue in a previous proof of \Cref{thm:lower-two-bits-two-prices}.
    
% \clearpage
% Bibliography
\bibliographystyle{plainnat}
\bibliography{references}

\clearpage
% Appendix
\appendix
% Appendix
\appendix

\section{Missing Proof}

\begin{claim}
    \label{claim:KL}
    For any time step $t$ and $i \in \{1,2\}$, the following inequality holds true for any $\eps \in (0,\tfrac 1{12})$:
    \[  
        \bno{ \mathbb P^i_{(Z_1, \dots, Z_{t-1} )} - \mathbb P^0_{(Z_1, \dots, Z_{t-1} )}}_{\mathrm{TV}} \le 4 \eps \sqrt{\Ez{N_T}}.
    \]
\end{claim}
\begin{proof}
    Recall, $\mathbb P^i_{(Z_1, \dots, Z_{t-1} )}$ is the (push-forward) probability measure induced by the first $t-1$ feedback observed on $(\{0,1\}^2)^{t-1}$.
    For any time step $t$ and $i=1, 2$ we can apply the Pinsker inequality and the chain rule for the KL divergence: 
    \begin{align}
    \nonumber
        \bno{ &\mathbb P_{(Z_1,\dots,Z_t)}^0 - \mathbb P^i_{(Z_1,\dots,Z_t)}}_{\mathrm{TV}}
    \le 
        \sqrt{\frac12 \kl \left( \Pb_{(Z_1,\dots,Z_t)}^0,\, \Pb^k_{(Z_1,\dots,Z_t)} \right)}
    \\
    \label{eq:secondKL}
    &\le
        \sqrt{ \frac12 \left(\kl\left( \Pb_{Z_1}^0, \, \Pb^k_{Z_1} \right) + \sum_{s=2}^t \Ez{{ \kl \left( \Pb_{Z_s \mid Z_1,\dots,Z_{s-1} }^0 , \, \Pb^k_{Z_s \mid Z_1,\dots,Z_{s-1}} \right)}}\right) }.
    \end{align}
    We upper bound the $\kl$ terms, starting from the one corresponding to the first time step. To this end, it is useful to partition the exploration region $X$ introduced in the main body into five parts: $X_1 = [0,\tfrac 16)\times [0,\tfrac 14] \cap \U$ (corresponding to $E_1$ in the main body), $X_2 = [0,\tfrac 16) \times (\tfrac 14,\tfrac 56]$, $X_3 = [0,\tfrac 16) \times (\tfrac 56, 1]$, $X_4 = [\tfrac 16,\tfrac 34) \times (\tfrac 56,1]$ and $X_5 = [\tfrac 34, 1] \times (\tfrac 56,1] \cap \U$ (corresponding to $E_2$ in the main body). For a visualization of the $X_i$ we refer to \Cref{fig:feedback_exploration}.

    The algorithm $\A$ we are considering is deterministic, thus $(P_1,Q_1)$ is a fixed pair of prices in $[0,1]^2$. We know that the only possibility for the feedback $Z_1$ to behave differently under $\Pb^i$ and $\Pb^0$ is to have $(P_1,Q_1)$ in $X$ (by \Cref{cl:no_exploration}), therefore we have the following:
    \begin{align}
    \kl&\big(\Pb^0_{Z_1},\Pb^i_{Z_1}\big)
        \nonumber
        \\
        \label{eq:KL1}
        &= \sum_{k = 1, \dots, 5}
        \left(\sum_{z \in \{0,1\}^2} \log \left( \frac{\Pz{Z_1 = z}}{\Pi{Z_1 = z}}\right)\Pz{Z_1 = z}\right)\ind{(P_1,Q_1) \in X_k}
    \end{align}
    We study separately the terms of \new{Equality~\ref{eq:KL1}} corresponding to different $k$. We have then $5$ cases.

    \paragraph{Case (i) : $(P_1,Q_1) \in X_1$} Under $\Pb^0$, $\Pb^1$ and $\Pb^2$ the probabilities of observing $(1,0)$ and $(0,0)$ are both $0$, because none of the eight points that make up the support of the distributions that lie in $Q^{(1,0)}_{P_1,Q_1}$ or $Q^{(0,0)}_{P_1,Q_1}$ (for the definition of the $Q^z_{p,q}$ we refer to the proof of \Cref{cl:no_exploration}). Under $\Pb^0$ the probabilities of observing $(1,1)$ is $\nicefrac 18$, while of observing $(0,1)$ is $\nicefrac 78$ (because $Q_{P_1,Q_1}^{(1,1)}$, respectively $Q_{P_1,Q_1}^{(0,1)}$ contains one point, respectively seven points, of the support of the distribution, and $\Pb^0$ is uniform over its support). Under $\Pb^{1}$, respectively $\Pb^{2}$, the probability of observing $(1,1)$ is $\nicefrac 18 + \eps$, respectively $\nicefrac 18 - \eps$, while of observing $(0,1)$ is $\nicefrac 78 - \eps$, respectively $\nicefrac 78 + \eps$.
    The corresponding term in \new{Equality~\ref{eq:KL1}} for $\Pb^1$ then becomes:
    \begin{align}
    \notag
         \sum_{z \in \{(1,1),(0,1)\}} &\log \frac{\Pz{Z_1 = z}}{ \Po{Z_1 = z} } \Pz{Z_1 = z} \cdot \ind{(P_1,Q_1) \in X_1}\\
            \notag
        &
        =
        \frac18\left[
        \log \frac{\nicefrac{1}{8}}{\nicefrac{1}{8} + \eps}
        +
        7\log\frac{\nicefrac{7}{8}}{\nicefrac{7}{8} - \eps}\right]
         \cdot \ind{(P_1,Q_1) \in X_1}\\
        \label{eq:KLterm1_1}&\le 16 \eps^2 \cdot \ind{(P_1,Q_1) \in X_1}, 
    \end{align}
    where the inequality holds true for any $\eps \in (0,\nicefrac 1{12})$ (in the following we avoid repeating that the inequality only holds in the domain specified in the statement). A similar chain of inequalities holds for probability measure $\Pb^2$:
    \begin{align}
    \notag
    \sum_{z \in \{(1,1),(0,1)\}} &\log \frac{\Pz{Z_1 = z}}{ \Pt{Z_1 = z} } \Pz{Z_1 = z} \cdot \ind{(P_1,Q_1) \in X_1}\\
        \nonumber
        & =
        \frac18\left[
        \log \frac{\nicefrac{1}{8}}{\nicefrac{1}{8} - \eps}
        +
        7\log\frac{\nicefrac{7}{8}}{\nicefrac{7}{8} + \eps}\right]
         \cdot \ind{(P_1,Q_1) \in X_1} \\
         \label{eq:KLterm1_2}
         &\le 16 \eps^2 \cdot \ind{(P_1,Q_1) \in X_1}, 
    \end{align}

    \begin{figure}
        \centering
        \includegraphics[width = 0.5\textwidth]{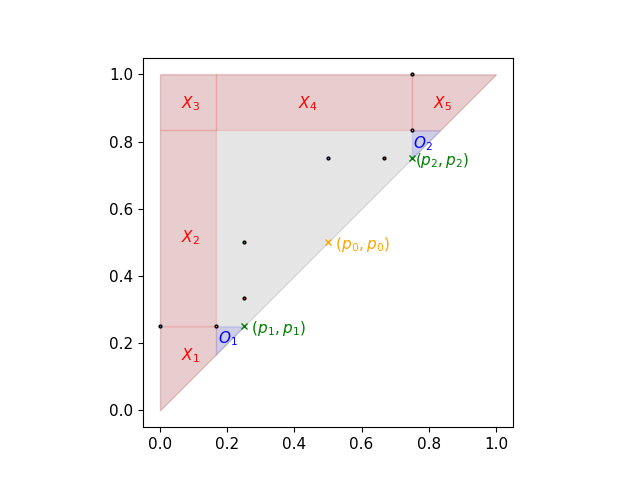}
        \caption{Feedback regions}
        \label{fig:feedback_exploration}
    \end{figure}
    
    \paragraph{Case (ii) : $(P_1,Q_1) \in X_2$} In $X_2$ we need to account for various cases. The feedback $(1,1)$ always happens with probability $0$, while $(1,0)$ is observed with probability $\nicefrac 18$ under $\Pb^0$, $\nicefrac 18 + \eps$ under $\Pb^1$, and $\nicefrac 18-\eps$ under $\Pb^2$. Feedback $(0,1)$ is observed with probability $\nicefrac 78 - \pi$ (regardless of the distribution), for some $\pi \in \{\nicefrac{1}8,\nicefrac{1}4,\nicefrac 38,\nicefrac{5}4\}$. Note, the parameter $\pi$ only depends on the value of the second coordinate, as an example if $(p,q)$ belongs to $X_2$ and it holds that $ \nicefrac 14 < q \le \nicefrac12$, then $Q_{p,q}^{(0,1)}$ contains $6$ points of the support, for a total probability (regardless of $i$) of $\nicefrac 34$ ($\pi = \nicefrac 18$). The last feedback, $(0,0)$ is observed with the remaining probability in the various cases.
    The corresponding term in \new{Equality~\ref{eq:KL1}} then becomes for $\Pb^1$:
    \begin{align}
    \notag
    \sum_{z \in \{(1,0),(0,0)\}} &\log \frac{\Pz{Z_1 = z}}{ \Po{Z_1 = z} } \Pz{Z_1 = z} \cdot \ind{(P_1,Q_1) \in X_2}\\
        \notag
        &=
        \left[\frac18
        \log \frac{\nicefrac{1}{8}}{\nicefrac{1}{8} + \eps}
        +
        \pi\log\frac{\pi}{\pi - \eps}\right]
         \cdot \ind{(P_1,Q_1) \in X_2}  \tag{for some $\pi \in \{\tfrac{1}8,\tfrac{1}4,\tfrac 38,\tfrac{5}4\}$}
        \\
        &
        \le
        \frac18\left[
        \log \frac{\nicefrac{1}{8}}{\nicefrac{1}{8} + \eps}
        +
        \log\frac{\nicefrac{1}{8}}{\nicefrac{1}{8} - \eps}\right]
         \cdot \ind{(P_1,Q_1) \in X_2} \tag{$\pi \to \pi \log \tfrac \pi{\pi-\eps}$ non-increasing in $(\eps,1)$} \\
      \label{eq:KLterm2_1}
      & \le 16 \eps^2 \cdot \ind{(P_1,Q_1) \in X_2}, 
    \end{align}
    where the inequality holds true for any $\eps \in (0,\nicefrac 1{24})$ and we used that the function $x \to x \log(\nicefrac x{x-\eps})$ is non-increasing in $(\eps,1)$, and $\pi$ is at least $\nicefrac 18$. Similar inequality can be derived for $\Pb^2$:
    \begin{align}
    \notag
    \sum_{z \in \{(1,0),(0,0)\}} &\log \frac{\Pz{Z_1 = z}}{ \Pt{Z_1 = z} } \Pz{Z_1 = z} \cdot \ind{(P_1,Q_1) \in X_2}\\
        \notag
        &=
        \left[\frac18
        \log \frac{\nicefrac{1}{8}}{\nicefrac{1}{8} - \eps}
        +
        \pi\log\frac{\pi}{\pi + \eps}\right]
         \cdot \ind{(P_1,Q_1) \in X_2}  \tag{for some $\pi \in \{\tfrac{1}8,\tfrac{1}4,\tfrac 38,\tfrac{5}4\}$}
        \\
        & \le
        \frac18\left[
        \log \frac{\nicefrac{1}{8}}{\nicefrac{1}{8} - \eps}
        +
        \log\frac{\nicefrac{1}{8}}{\nicefrac{1}{8}+ \eps}\right]
         \cdot \ind{(P_1,Q_1) \in X_2} \tag{$\pi \to \pi \log \tfrac \pi{\pi+\eps}$ non-increasing in $(0,1)$}\\
         \label{eq:KLterm2_2}
         &\le 16 \eps^2 \cdot \ind{(P_1,Q_1) \in X_2}. 
    \end{align}
    
    \paragraph{Case (iii) : $(P_1,Q_1) \in X_3$}  Under $\Pb^0$, $\Pb^1$ and $\Pb^2$ the probability of observing $(1,1)$ are both $0$, while that of observing $(0,0)$ are both $\nicefrac 34$. Under $\Pb^0$ the probabilities of observing $(1,0)$ or $(0,1)$ are both $\nicefrac 18$, while for $\Pb^i$ are $\nicefrac 18 + \eps$ and $\nicefrac 18 - \eps$ (the order depends on $i = 1,2$). The corresponding term in \new{Equality~\ref{eq:KL1}} then becomes:
    \begin{align}
     \notag
    \sum_{z \in \{(1,0),(0,1)\}} &\log \frac{\Pz{Z_1 = z}}{ \Pi{Z_1 = z} } \Pz{Z_1 = z} \cdot \ind{(P_1,Q_1) \in X_3}\\ \notag
&
=
    \frac18\left[
    \log \frac{\nicefrac{1}{8}}{\nicefrac{1}{8} - \eps}
    +
    \log\frac{\nicefrac{1}{8}}{\nicefrac{1}{8} + \eps}\right]
     \cdot \ind{(P_1,Q_1) = X_3} \\
     \label{eq:KLterm3}
&\le 16 \eps^2 \cdot \ind{(P_1,Q_1) \in X_3}.
\end{align}

\paragraph{Case (iv) : $(P_1,Q_1) \in X_4$} In $X_4$ we need to account for various cases, symmetrically to what we did for $X_2$. The feedback $(1,1)$ happens with probability $0$ in any case, while $(0,1)$ is observed with probability $\nicefrac 18$ under $\Pb^0$, $\nicefrac 18 - \eps$ under $\Pb^1$, and $\nicefrac 18 + \eps$ under $\Pb^2$. Feedback $(1,0)$ is observed with probability $\nicefrac 78 - \pi$ (regardless of the distribution), for some $\pi \in \{\nicefrac{1}8,\nicefrac{1}4,\nicefrac 38,\nicefrac{5}8\}$; note, the parameter $\pi$ only depends on the value of the first coordinate. The last feedback, $(0,0)$ is observed with the remaining probability in the various cases. The corresponding term in \new{Equality~\ref{eq:KL1}} can be bounded as in \new{Inequality~\ref{eq:KLterm2_2}} for $\Pb^1$:
    \begin{align}
     \notag
    \sum_{z \in \{(0,0),(0,1)\}} &\log \frac{\Pz{Z_1 = z}}{ \Po{Z_1 = z} } \Pz{Z_1 = z} \cdot \ind{(P_1,Q_1) \in X_4}\\
        &=
        \left[\frac18
        \log \frac{\nicefrac{1}{8}}{\nicefrac{1}{8} - \eps}
        +
        \pi\log\frac{\pi}{\pi + \eps}\right]
         \cdot \ind{(P_1,Q_1) \in X_4} 
        \nonumber\\
        \label{eq:KLterm4_1}
        &\le 16 \eps^2 \cdot \ind{(P_1,Q_1) \in X_4}. 
    \end{align}
    The same derivation as in \new{Inequality~\ref{eq:KLterm2_1}} can be carried over for $\Pb^2$:
    \begin{align}
    \notag
    \sum_{z \in \{(0,0),(0,1)\}} &\log \frac{\Pz{Z_1 = z}}{ \Pt{Z_1 = z} } \Pz{Z_1 = z} \cdot \ind{(P_1,Q_1) \in X_4}\\
        \notag &=
        \left[\frac18
        \log \frac{\nicefrac{1}{8}}{\nicefrac{1}{8} + \eps}
        +
        \pi\log\frac{\pi}{\pi - \eps}\right]
         \cdot \ind{(P_1,Q_1) \in X_4}
        \\
        \label{eq:KLterm4_2}
        &\le 16 \eps^2 \cdot \ind{(P_1,Q_1) \in X_4}.
    \end{align}

\paragraph{Case (v) : $(P_1,Q_1) \in X_5$} Under $\Pb^0$, $\Pb^1$ and $\Pb^2$ the probabilities of observing $(0,1)$ and $(0,0)$ are both $0$, because none of the eight points that make up the support $\supp $ of the distributions lies in $Q^{(0,1)}_{P_1,Q_1}$ or $Q^{(0,0)}_{P_1,Q_1}$. Under $\Pb^0$ the probabilities of observing $(1,1)$ is $\nicefrac 18$, while of observing $(1,0)$ is $\nicefrac 78$ (because $Q_{P_1,Q_1}^{(1,1)}$, respectively $Q_{P_1,Q_1}^{(1,0)}$ contains one point, respectively seven points, of the support of the distribution, and $\Pb^0$ is uniform over its support). Under $\Pb^{1}$, respectively $\Pb^{2}$, the probability of observing $(1,1)$ is $\nicefrac 18 - \eps$, respectively $\nicefrac 18 + \eps$, while of observing $(0,1)$ is $\nicefrac 78 + \eps$, respectively $\nicefrac 78 - \eps$.
    The corresponding term in \new{Equality~\ref{eq:KL1}} for $\Pb^1$ then becomes:
    \begin{align}
    \notag
    \sum_{z \in \{(1,1),(0,1)\}} &\log \frac{\Pz{Z_1 = z}}{ \Po{Z_1 = z} } \Pz{Z_1 = z} \cdot \ind{(P_1,Q_1) \in X_5}\\
    \notag
        &=
        \frac18\left[
        \log \frac{\nicefrac{1}{8}}{\nicefrac{1}{8} - \eps}
        +
        7\log\frac{\nicefrac{7}{8}}{\nicefrac{7}{8} + \eps}\right]
         \cdot \ind{(P_1,Q_1) \in X_5} \\
         \label{eq:KLterm5_1}
         &\le 16 \eps^2 \cdot \ind{(P_1,Q_1) \in X_5}. 
    \end{align}
    A similar chain of inequalities holds for $\Pb^2$:
    \begin{align}
    \notag
        \notag
    \sum_{z \in \{(1,1),(0,1)\}} &\log \frac{\Pz{Z_1 = z}}{ \Pt{Z_1 = z} } \Pz{Z_1 = z} \cdot \ind{(P_1,Q_1) \in X_5}\\
    \notag
        &=
        \frac18\left[
        \log \frac{\nicefrac{1}{8}}{\nicefrac{1}{8} + \eps}
        +
        7\log\frac{\nicefrac{7}{8}}{\nicefrac{7}{8} - \eps}\right]
         \cdot \ind{(P_1,Q_1) \in X_5} \\
         \label{eq:KLterm5_2}
         &\le 16 \eps^2 \cdot \ind{(P_1,Q_1) \in X_5}.
    \end{align}

    We can plug into \new{Equality~\ref{eq:KL1}} the case analysis  (\new{Inequalities~\ref{eq:KLterm1_1}, \ref{eq:KLterm2_1}, \ref{eq:KLterm3}, \ref{eq:KLterm4_1}, and \ref{eq:KLterm5_1}} for $\Pb^1$ 
 and \new{Inequalities~\ref{eq:KLterm1_2}, \ref{eq:KLterm2_2}, 
 \ref{eq:KLterm3}, \ref{eq:KLterm4_2}, and \ref{eq:KLterm5_2}} for $\Pb^2$) to get the desired bound for the first $\mathrm{KL}$ term:
    \begin{equation}
    \label{eq:KLbound1}    
        \kl\big(\Pb^0_{Z_1},\Pb^i_{Z_1}\big)  \le 16 \eps^2 \ind{(P_1,Q_1) \in X}
    \end{equation}

        Once we have our argument for $s = 1$ we generalize it for the other time steps. The crucial observation is, once again, that the algorithm $\A$ is deterministic, therefore if we condition on a specific realization of the feedback signals up to time $s-1$, the price $(P_s,Q_s)$ posted at time $s$ becomes deterministic. Moreover, the random variables $(S_s,Q_s)$ are drawn independently at each round, therefore for any feedback $z$ and time $s$ it holds that:
\begin{align*}
&
    \kl (\Pb_{Z_{s} \mid Z_1,\dots,Z_{s-1}}^0 , \, \Pb^i_{Z_s \mid Z_1,\dots,Z_{s-1}} )
\\ &= \sum_{\substack{k=1\dots 5 \\ z \in \{0,1\}^2}}  \log \frac{\Pz{Z_s = z \mid Z_{1 : s-1}}}{\Pi{Z_s = z\mid Z_{1 : s-1}}}\Pz{Z_1 = z\mid Z_{1 : s-1}}\ind{(P_s,Q_s) \in X_k \mid Z_{1 : s-1}}
\\
&
=\sum_{\substack{k=1\dots 5 \\ z \in \{0,1\}^2}} \log \frac{\Pz{Z_s = z \mid(P_s,Q_s)  \in X_k }}{\Pi{Z_s = z \mid(P_s,Q_s)  \in X_k }}\Pz{Z_s = z \mid(P_s,Q_s)  \in X_k }\ind{(P_s,Q_s) \in X_k \mid Z_{1 : s-1}},
\end{align*}
where we denoted with $Z_{1 : s-1}$ the sequence of random variables $Z_1, Z_2, \dots Z_{s-1}.$
With the same exact case analysis and calculations that has been carried over for $s=1$ (this time on $(P_s,Q_s)$ given past history and not $(P_1,Q_1)$) we get the same bound as in \new{Inequality~\ref{eq:KLbound1}}:
\begin{align}
    \label{eq:KLbound2}
    \kl (\Pb_{Z_s \mid Z_1,\dots,Z_{s-1}}^0 , \, \Pb^i_{Z_s \mid Z_1,\dots,Z_{s-1}} ) \le 16 \eps^2\ind{(P_s,Q_s) \in X \mid Z_1 \dots  Z_{s-1}}.
\end{align}
Finally, we can plug the bounds in \new{Inequalities~\ref{eq:KLbound1} and \ref{eq:KLbound2}} into our original inequality (\new{Inequality~\ref{eq:secondKL}}) to get the final result 
\begin{align*}
    \nonumber
        &\bno{ \mathbb P_{(Z_1,\dots,Z_t)}^0 - \mathbb P^i_{(Z_1,\dots,Z_t)}}_{\mathrm{TV}}
    \\
    \nonumber
    &\quad \le 4 \eps \sqrt{\ind{(P_1,Q_1) \in X} +  \sum_{s=2}^t \Ez{\ind{(P_s,Q_s) \in X_k \mid Z_1 \dots Z_{s-1}}}} \\
    &\quad = 4 \eps \sqrt{\Ez{N_T}}.\qedhere
\end{align*}
\end{proof}

\end{document}